\documentclass[journal]{IEEEtran}
\newcommand{\maximize}{\mathop{\textrm{maximize}}}

\newtheorem{definition}{Definition}
\newtheorem{theorem}{Theorem}

\usepackage{color, colortbl}
\usepackage{graphicx}
\usepackage{framed}
\usepackage{amssymb}
\definecolor{LightCyan}{rgb}{0.88,1,1}
\definecolor{Gray}{gray}{0.6}

\usepackage{cite}
\ifCLASSINFOpdf
\else
\fi
\usepackage[cmex10]{amsmath}
\usepackage{algpseudocode}
\usepackage{algorithm}
\usepackage[table]{xcolor}
\usepackage{multirow}
\begin{document}
\title{Real-Time Dynamic Spectrum Management for Multi-User Multi-Carrier Communication Systems}

\newcommand{\argmax}{\operatornamewithlimits{argmax}}


\author{Paschalis~Tsiaflakis\thanks{P. Tsiaflakis is a postdoctoral fellow funded by the Research Foundation-Flanders (FWO).
This research work was carried out at the ESAT Laboratory of KU Leuven, in the frame of KU Leuven Research Council CoE  PFV/10/002 (OPTEC), KU Leuven Research Council Bilateral Scientific Cooperation Project Tsinghua University 2012-2014,
the Belgian Programme on Interuniversity Attraction Poles initiated by the
Belgian Federal Science Policy Office IUAP P7/23 `Belgian network on stochastic modeling analysis design and optimization of communication systems' (BESTCOM) 2012-2017,
Concerted Research Action GOA-MaNet, Research Project FWO nr. G.091213 'Cross-layer optimization with real-time adaptive dynamic spectrum management for fourth generation broadband access networks'.
The scientific responsibility is assumed by its authors.}\thanks{P. Tsiaflakis and M. Moonen are with the Department of Electrical Engineering (ESAT), STADIUS Center for Dynamical Systems, Signal Processing and Data Analytics, KU Leuven, Kasteelpark Arenberg 10 bus 2446, B-3001 Leuven. (email: Paschalis.Tsiaflakis@esat.kuleuven.be; Marc.Moonen@esat.kuleuven.be)}\thanks{F. Glineur is affiliated with (1) the Center for Operations Research and Econometrics, Universit\'e catholique de Louvain, B-1348 Louvain-la-Neuve, Belgium, and (2) Information and Communication Technologies, Electronics and Applied Mathematics Institute, Universit\'e catholique de Louvain, B-1348 Louvain-la-Neuve, Belgium. (email: Francois.Glineur@uclouvain.be)}\thanks{This work has been submitted to the IEEE for possible publication. Copyright may be transferred without notice, after which this version may no longer be accessible.}, \IEEEmembership{Member, IEEE}, Fran\c{c}ois~Glineur, and~Marc~Moonen, \IEEEmembership{Fellow, IEEE}}

\maketitle

\begin{abstract}
%\boldmath
Dynamic spectrum management is recognized as a key technique to tackle interference
in multi-user multi-carrier communication systems and networks. 
However existing dynamic spectrum management algorithms may not be suitable when
the available computation time and compute power are limited, i.e., when
a very fast responsiveness is required. In this paper, we present 
a new paradigm, theory and algorithm for real-time dynamic spectrum management (RT-DSM)
under tight real-time constraints. Specifically, a RT-DSM algorithm can 
be stopped at any point in time while guaranteeing a feasible and improved solution.
This is enabled by the introduction of a novel difference-of-variables (DoV) 
transformation and problem reformulation, for
which a primal coordinate ascent approach is proposed with exact line search
via a logarithmicly scaled grid search. The concrete proposed algorithm is referred to
as iterative power difference balancing (IPDB). 
Simulations for different realistic wireline and wireless interference limited systems
demonstrate its good performance, low complexity and wide applicability 
under different configurations.
\end{abstract}


%\begin{center}
%{\small \textbf{EDICS}: NET-PHYL, NET-RSMG, NET-SYLO, OPT-NCVX, SPC-DSLP, SPC-INTF, SPC-MULT}
%\end{center}

% \fi
\IEEEpeerreviewmaketitle

% \section{Introduction}
% for IEEE journal papers produced under \LaTeX\ using
% IEEEtran.cls version 1.7 and later.
% I wish you the best of success.

 
%\hfill January 11, 2007

\section{Introduction}\label{sec:introduction}
Interference is a key performance-limiting factor in many
state-of-the-art communication systems and networks \cite{Gesbert2010,Lee2012,Zheng2012,Lee2012b,DSM_Song,vectoredTrans,raphaelthesis}. 
In particular, when multiple users transmit simultaneously in 
a common frequency bandwidth, significant interference levels
can be observed among them in practical systems. This can result in large data rate reductions \cite{raphaelthesis,vectoredTrans,Maes2013}, 
poor spectral and energy efficiency \cite{Mao2013,energyDSLnordstrom,Tsiaflakis2009a,greenercopper, Hooghe2011, Monteiro2009, Wolkerstorfer2012a}, unstable behaviour due to transient 
interference \cite{Panigrahi2006,Jagannathan2008,Sarhan2012,Ginis2005}, unfairness due to unbalanced interference impact \cite{Tsiaflakis2011} and other 
performance degradations.

Dynamic spectrum management (DSM) is recognized as an important technique
to tackle these performance degradations in such interference limited systems
\cite{DSM_Song,vectoredTrans,raphaelthesis,Maes2013,dsb,energyDSLnordstrom,Tsiaflakis2009a,greenercopper, 
Hooghe2011, Monteiro2009, Wolkerstorfer2012a,Panigrahi2006,Jagannathan2008,Sarhan2012,Ginis2005,
Tsiaflakis2011,Statovci2006,Forouzan2012,Moraes2013,Leung2013,Huberman09, Wang2012, dual_journal,DSMluo,OSB,
ISB_Raphael,ISB_WeiYu,tsiaflakis_bbosb,Weeraddana2011,Wolkerstorfer2012b,ASB,Son2011,iterativeWaterfilling,Huberman2012,Luo2006,scalejournal,MIW}.
In digital subscriber line (DSL) literature,
DSM is typically categorized into three levels, namely DSM 1, DSM 2 and DSM 3.
DSM~1 corresponds to single user management in terms of impulse-noise control,
delay parameter tuning and transmit spectrum shaping \cite{scholarpediaDSL}. DSM 2 addresses 
solutions where the transmit spectra of all users are jointly managed \cite{raphaelthesis,DSM_Song,Huberman09}
so as to prevent the destructive impact of interference.
DSM 3 is also referred to as vectoring and consists of the application of 
signal coordination methods that can actively cancel interference between users \cite{vectoredTrans,raphaelthesis,Maes2013,Moraes2013,Leung2013}. 
We want to briefly highlight here that the word 'dynamic' in DSM does not refer
to time dynamic adaptation of the spectrum management resources, but 
rather to the adaptation of the spectrum management resources taking the 
concrete physical channel conditions of the considered scenario into account.

In this paper the focus is on DSM through the management of transmit spectra for general 
multi-user multi-carrier systems, including wireline DSL systems (corresponding 
to DSM 2) as well as wireless interference limited systems. Here, the transmit spectra of 
all users in the system are jointly managed and optimized , where each user employs a 
multi-carrier transmission technique such as orthogonal frequency division multiplexing 
(OFDM) or discrete multitone (DMT). In the remainder of this paper, we will refer to this technique shortly by DSM, 
as this term is similarly used in other literature \cite{raphaelthesis,DSMluo,Huberman09}.
Furthermore we follow a standard interference channel system model where 
the interference is treated as additive white Gaussian (AWG) noise, which is a 
very common practical model in operational networks \cite{raphaelthesis}.

Research work on DSM has progressed significantly over 
the last decade. More specifically, a whole range of DSM 
algorithms has been proposed ranging from centralized \cite{OSB,ISB_Raphael,ISB_WeiYu,tsiaflakis_bbosb,Weeraddana2011,Wolkerstorfer2012b}, 
to distributed \cite{dsb,Ren2010,scalejournal,Chiang2008,MIW} and autonomous algorithms \cite{ASB,dsb,Son2011,iterativeWaterfilling,Huberman2012,Luo2006}. Each 
  of these has its specific properties in terms of computational complexity 
and level of coordination. We refer to \cite{Huberman09,dsb} and references 
therein for a comparison and an overview of DSM algorithms proposed 
in DSL literature. The DSL setting represents one 
relevant example of a multi-user multi-carrier interference channel. 
However DSM is also of interest in several wireless settings, where 
similar algorithms have been proposed. Examples are multicell downlink 
DSM or inter-cell interference coordination for heterogeneous networks \cite{Son2011}, 
DSM for multi-user multi-channel cellular relay 
networks \cite{Ren2010}, and OFDM based cognitive radio systems \cite{Bansal2008,Mao2013}.

However, none of the previously proposed DSM algorithms 
have addressed real-time computation constraints. More specifically, when 
computation time and compute power are limited, there is no guarantee that 
a suitable solution can be found with existing DSM algorithms. This is because existing 
DSM algorithms typically follow an iterative approach where it is not known 
in advance how many iterations
are required to converge to a feasible and good solution. Furthermore existing 
DSM algorithms typically follow a dual decomposition approach where solution 
feasibility and good performance is not guaranteed until after convergence. 
In addition, an important issue of tackling the nonconvex DSM 
problem in the dual domain is the possible non-zero duality gap as the number of 
subcarriers used in the multi-carrier transmission is finite \cite{dual_journal, DSMluo}.

Our focus in this paper is to tackle the above issues by a new 
paradigm and theory of \emph{real-time dynamic spectrum management}. The corresponding 
real-time dynamic spectrum management algorithms 
succeed in working under real-time constraints where the computation 
time and the compute power are limited. This property is highly 
desirable when real-time responsiveness to changes 
in the network is to be guaranteed, such as changing channels and noises, users joining or 
leaving the network, changing QoS requirements, crosslayer control, etc. 
To the best of our knowledge, there exists no literature on resource 
allocation for interference limited communication systems that addresses 
such real-time constraints. 

To enable this new paradigm, we first propose a novel transformation, referred to as the \emph{difference-of-variables} 
(DoV) transformation, which transforms the standard DSM 
problem into a problem with alternative primal variables, referred to as 
power difference variables. With this reformulation in hand, a first real-time 
DSM algorithm is proposed, which is referred to as \emph{iterative power difference 
balancing (IPDB)}. This algorithm combines the DoV problem
reformulation with a solution that follows a coordinate ascent approach
with exact line search via a logarithmicly scaled grid search. The combination
of these two ingredients results in an efficient algorithm for which the effectiveness
and real-time behaviour are analyzed and evaluated for different settings.

The paper is organized as follows. Section~\ref{sec:systemmodel} briefly 
describes the multi-user multi-carrier system model and DSM. Section~\ref{sec:realtime} first gives a definition of real-time 
DSM. Then the DoV transformation and problem are proposed. Finally, the IPDB algorithm is 
presented. This is extended with a procedure for 
dealing with inequality power constraints and equalization. The 
performance for different wireline and wireless scenarios and settings 
is presented in Section~\ref{sec:simulations}.

\section{System model and dynamic spectrum management}\label{sec:systemmodel}

We consider a multi-user communication system with a set $\mathcal{N}=\{1,\ldots,N\}$ of $N$ 
communication links over a common frequency band. Each link consists of
a transmitter-receiver pair, and is also referred to as a user. 
In addition, each user employs a multi-carrier transmission scheme,
such as OFDM or DMT. We assume perfect synchronization and a cyclic prefix length longer than the maximum channel length 
(considering direct as well as interference channels), so that the user 
data are transmitted independently and in parallel on the different subcarriers, 
also referred to as tones. The set of $K$ tones is denoted as $\mathcal{K}=\{1,\ldots,K\}$.
All users can transmit on all tones, resulting in overlapping
transmit spectra and thus multi-user interference. Note that our system also
includes the single-user case, i.e., with $N=1$, as a special case.

We focus on dynamic spectrum management through multi-user multi-carrier
transmit power management and optimization. No signal coordination or vectoring
between transmitters or receivers is assumed. Each user 
thus employs a single-user decoder. This case is well in line with many practical settings where a distinct physical location
or a limited communication between transmitters and receivers does not allow for 
signal coordination.  
We follow the common standard interference channel system model where the multi-user interference
is treated as AWG noise. Perfect channel state information is assumed at transmitters and receivers.
The achievable bit rate of user $n$ on tone $k$ is then given as follows
\begin{equation}\label{eq:bitrate}
 b_k^n(\mathbf{s}_k) \triangleq \log_2 \left( 1+ \frac{s_k^n}{{\displaystyle \sum_{m \neq n}} a_k^{n,m} s_k^m +
z_k^n} \right),
\end{equation}
with $\mathbf{s}_k=[s_k^1,\ldots,s_k^N]^T$, $s_k^n$ denoting the transmit power of user $n$ on tone $k$, 
$a_k^{n,m}$ denoting the normalized channel gain from user $m$ to
user $n$ on tone $k$, and $z_k^n$ denoting the normalized received 
noise power for user $n$ on tone $k$.
A signal to noise ratio (SNR) gap \cite{UnderstandingDSL} that characterizes imperfect coding 
and signal modulation, and a noise margin, may be included in the 
normalized channel gains and noise power.

The DSM problem can then be formulated as follows
\begin{equation}\label{eq:smp}
\begin{array}{cl}
{\displaystyle \maximize_{s_k^n, k \in \mathcal{K}, n \in \mathcal{N}}} & 
{\displaystyle \sum_{n \in \mathcal{N}} w_n R^n(\mathbf{s}_k, k \in \mathcal{K})}\\
\mathrm{subject~to} & {\displaystyle P^n(s_k^n, k \in \mathcal{K}) = P^{n,\mathrm{tot}}}, \forall n \in \mathcal{N}\\
& 0 \leq s_k^n \leq s_k^{n,\mathrm{mask}}, \forall n \in \mathcal{N}, \forall k \in \mathcal{K},\\
& \\
\mathrm{with} & \left\{ 
\begin{array}{l}
{\displaystyle R^n(\mathbf{s}_k, k \in \mathcal{K}) \triangleq  \sum_{k \in \mathcal{K}}
b_k^n(\mathbf{s}_k)}\\ 
{\displaystyle P^n(s_k^n, k \in \mathcal{K}) \triangleq  \sum_{k \in \mathcal{K}} s_k^n},
\end{array} \right.\\
\end{array} 
\end{equation}
with $R^n(\mathbf{s}_k, k \in \mathcal{K})$ denoting the achievable data rate for user $n$ and its corresponding
weighting $w_n$, $P^n(s_k^n, k \in \mathcal{K})$ denoting the total allocated (transmit) power of user $n$, 
the constant $P^{n,\mathrm{tot}}$ denoting the total power budget for user $n$,
and the constant $s_k^{n,\mathrm{mask}}$ denoting the maximum transmit power 
(spectral mask) of user $n$ on tone $k$. This corresponds to a maximization of the sum of 
the weighted achievable data rates (with multiple tones), under per-user
total power constraints and per-tone spectral masks. 

The transmit spectrum of a user refers to the user's transmit
power on all the tones. These transmit spectra are the optimization variables for
the DSM problem.

We want to highlight that the per-user total power constraints are expressed 
as equality constraints. This is further extended to inequality constraints
in Section~\ref{sec:equalitymasks}.

\section{Real-Time Dynamic Spectrum Management}\label{sec:realtime}

\begin{figure*}
\begin{equation}\label{eq:rtsmp}
\begin{array}{cl}
{\displaystyle \maximize_{t_k^n, k \in \mathcal{K}, n \in \mathcal{N}}} & 
{\displaystyle \sum_{n \in \mathcal{N}} w_n \sum_{k \in \mathcal{K}}
\log_2 \left( 1+ \frac{{\displaystyle \sum_{j \in \mathcal{K}} \beta_k^n(j) t_j^n + P^{n,\mathrm{tot}} \gamma_k^n}}
{{\displaystyle \sum_{m \neq n} a_k^{n,m} \left( \sum_{j \in \mathcal{K}} \beta_k^m(j) t_j^m + P^{m,\mathrm{tot}} \gamma_k^m\right) +
z_k^n}} \right)}\\
\mathrm{subject~to} & {\displaystyle 0 \leq \sum_{j \in \mathcal{K}} \beta_k^n(j) t_j^n+P^{n,\mathrm{tot}} \gamma_k^n \leq s_k^{n,\mathrm{mask}}, \forall n \in \mathcal{N}, \forall k \in \mathcal{K}}
\end{array} 
\end{equation}
\end{figure*}

Real-time computation is an important 
challenge in practice where computation time and compute power 
of communication devices and systems are limited.
In this section, we present a new paradigm and theory for \emph{real-time dynamic spectrum management} (RT-DSM).
We first introduce our definition of a RT-DSM algorithm
in Section~\ref{sec:RTdefinition}. To enable the design of RT-DSM algorithms, we then
propose a novel transformation, also referred to as a difference-of-variables transformation, in Section~\ref{sec:dvtransformation}. Using this transformation, we reformulate
the DSM problem in terms of power difference variables. This
allows for the design of a first RT-DSM algorithm in Section~\ref{sec:idb},
referred to as \emph{iterative power difference balancing}. 
This is further extended with a procedure for dealing 
with inequality per-user total power constraints in Section~\ref{sec:equalitymasks}, and
an equalization procedure to tackle non-smooth solution behaviour in Section~\ref{sec:equalization}.

\subsection{Definition of Real-Time Dynamic Spectrum Management Algorithm}\label{sec:RTdefinition}

To provide a concrete label and definition of the algorithms targeted in 
this paper, we introduce the following definition:
\begin{framed}
\hspace{-0.1cm}\begin{definition} 
 \textbf{[Real-time dynamic spectrum management (RT-DSM) algorithm]} A real-time dynamic spectrum management algorithm sequentially updates the 
 transmit power variables such that these satisfy all constraints after each 
 update.
\end{definition}\label{def:rtso}
\end{framed}

This definition implies that RT-DSM algorithms can be stopped after 
each update (even after a single update of any transmit power variable), 
and as such they are suitable for execution under very tight computation
time and compute power budgets, as they can be stopped whenever one of both 
resources is exhausted. This guarantees fast responsiveness, and allows
for real-time operation. 

\subsection{Difference-of-Variables (DoV) Transformation and Optimization}\label{sec:dvtransformation}

The original DSM problem (\ref{eq:smp}) consists of a
separable objective function and coupled per-user total power constraints. An 
important step towards the design of RT-DSM algorithms is to eliminate the
per-user total power constraints. To enable this we propose to replace the 
primal variables $s_k^n~\forall n,k$ with an alternative set of primal variables $t_k^n~\forall n,k$,
where the latter will be referred to as the \emph{power difference variables}.
For this we propose a novel transformation of variables, referred to as the \emph{difference-of-variables
(DoV) transformation}:
\begin{eqnarray}
 & & s_k^n = \sum_{j \in \mathcal{K}} \beta_k^n(j) t_j^n + P^{n,\mathrm{tot}} \gamma_k^n, \quad n \in \mathcal{N}, k \in \mathcal{K} \label{eq:transformation}\\
 & \mathrm{with} & \sum_{k \in \mathcal{K}} \beta_k^n(j) = 0, \quad n \in \mathcal{N}, j \in \mathcal{K}\label{cond:difference}\\
 & & \sum_{k \in \mathcal{K}} \gamma_k^n = 1, \quad n \in \mathcal{N}\label{cond:equalpower}\\
 & & \beta_k^n(k) > 0, \quad n \in \mathcal{N}, k \in \mathcal{K}\label{cond:posdirect}
\end{eqnarray}
where $\beta_k^n(j), \gamma_k^n$ are (fixed) arbitrary constants that can take any
value satisfying constraints (\ref{cond:difference}), (\ref{cond:equalpower}) and (\ref{cond:posdirect}).
We also define the following set $\mathcal{B}_k^n$ for later use
 \begin{eqnarray}
  \mathcal{B}_k^n & = & \left\{j \in \mathcal{K} | \beta_k^n(j) \neq 0 \right\},
 \end{eqnarray}
which denotes the set of tones for user $n$ and tone $k$ with power difference 
variables that influence $s_k^n$. 

Using the DoV transformation (\ref{eq:transformation}), we obtain a reformulation of 
(\ref{eq:smp}) as given in the following theorem:
\begin{theorem}
 Applying a DoV transformation (\ref{eq:transformation}), satisfying 
 constraints (\ref{cond:difference}) (\ref{cond:equalpower}) (\ref{cond:posdirect}),
 to the DSM problem (\ref{eq:smp}) results in the 
 reformulated problem (\ref{eq:rtsmp}).
\end{theorem}

 \begin{proof}
 The objective and constraints of (\ref{eq:rtsmp}) can be obtained by applying the DoV transformation
 to the objective and per-tone spectral mask constraints of (\ref{eq:smp}).
 The per-user total power constraints of (\ref{eq:smp}) are not present anymore in the reformulation (\ref{eq:rtsmp}).
 This is because the proposed DoV transformation (\ref{eq:transformation})
 ensures that these constraints are satisfied for all values of the power difference
 variables $t_k^n$. This can straightforwardly be proven as follows:
 \begin{displaymath}
  \begin{array}{cl}
   {\displaystyle \sum_{k \in \mathcal{K}} s_k^n} & = {\displaystyle \sum_{k \in \mathcal{K}} \left( \sum_{j \in \mathcal{K}} \beta_k^n(j) t_j^n + P^{n,\mathrm{tot}} \gamma_k^n \right)} \\
   & = {\displaystyle\left( \sum_{j \in \mathcal{K}} t_j^n \underbrace{\sum_{k \in \mathcal{K}} \beta_k^n(j)}_{=0} \right) + \left( P^{n,\mathrm{tot}} 
   \underbrace{\sum_{k \in \mathcal{K}} \gamma_k^n}_{=1} \right)}\\
   & = P^{n,\mathrm{tot}}
  \end{array}
 \end{displaymath} 
 \end{proof}
 
 \begin{figure*}
    \begin{equation}\label{eq:1Drtsmp}
  \begin{array}{cl}
     {\displaystyle \maximize_{t_k^n}} & f_k^n(t_k^n)\\
 \mathrm{subject~to} & {\displaystyle t_k^{n,\mathrm{min}} \leq t_k^n \leq t_k^{n,\mathrm{max}}}\\
  \mathrm{with} & \left\{ \begin{array}{ll} f_k^n(t_k^n) = {\displaystyle \sum_{k \in \mathcal{A}_k^n} \sum_{n=1}^N w_n 
     \log_2 \left( 1 + \frac{{\displaystyle \sum_{j \in \mathcal{B}_k^n} \beta_k^n(j) t_j^n + P^{n,\mathrm{tot}} \gamma_k^n}}{{\displaystyle \sum_{m \neq n} a_k^{n,m} \left( \sum_{j \in \mathcal{B}_k^m} \beta_k^m(j) t_j^m + P^{m,\mathrm{tot}} \gamma_k^m\right) +
z_k^n}} \right)}\\ 
  x = \max \{ -\frac{1}{\beta_q^n(k)} \left( \sum_{j \in \mathcal{B}_q^n \backslash \{k\}} \beta_q^n(j) t_j^n+P^{n,\mathrm{tot}}\gamma_q^n \right), 
   \forall q \in \mathcal{A}_k^n ~ \& ~ \beta_q^n(k) > 0 \}\\
   y = \max \{ \frac{1}{\beta_q^n(k)} \left(s_q^{n,\mathrm{mask}} - \left(\sum_{j \in \mathcal{B}_q^n \backslash \{k\}} \beta_q^n(j) t_j^n+P^{n,\mathrm{tot}} \gamma_q^n\right) \right), 
   \forall q \in \mathcal{A}_k^n ~ \& ~ \beta_q^n(k) < 0  \}\\
   t_k^{n,\mathrm{min}}= \max \{x,y\}\\
   u = \min \{ \frac{1}{\beta_q^n(k)} \left(s_q^{n,\mathrm{mask}} - \left( \sum_{j \in \mathcal{B}_q^n \backslash \{k\}} \beta_q^n(j) t_j^n+P^{n,\mathrm{tot}} \gamma_q^n \right) \right), 
   \forall q \in \mathcal{A}_k^n ~ \& ~ \beta_q^n(k) > 0  \}\\
   v = \min \{ -\frac{1}{\beta_q^n(k)} \left( \sum_{j \in \mathcal{B}_q^n \backslash \{k\}} \beta_q^n(j) t_j^n+P^{n,\mathrm{tot}} \gamma_q^n \right), 
   \forall q \in \mathcal{A}_k^n ~ \& ~ \beta_q^n(k) < 0  \}\\
   t_k^{n,\mathrm{max}} = \min \{ u,v\} \end{array} \right.\\
  \end{array}
  \end{equation}
 \end{figure*}
 
 The strength of reformulation (\ref{eq:rtsmp}) is that the coupled per-user
 total power constraints are eliminated, and that the reformulation is 
 expressed in terms of the power difference
 variables $t_k^n$. Because of the constraint (\ref{cond:difference}), 
 each power difference variable $t_k^n$ adds some transmit power to some tones 
 but subtracts the same amount of transmit power from other tones, resulting in a zero total power change
 operation. This is also the reason why $t_k^n$ is referred to as a power difference variable.
 The above properties of the reformulated DSM problem (\ref{eq:rtsmp})
 are crucial to enable the design of RT-DSM algorithms as will be shown in Section~\ref{sec:idb}.
 
 Reformulation (\ref{eq:rtsmp}) displays coupling in both the objective as well as the constraints. However,
 the coupling can be of much smaller size (compared to (\ref{eq:smp})) in the sense that it couples only a subset of all tones. 
 More specifically, the bit loading $b_k^n$ on a given tone and
 the spectral mask constraints are impacted by a number of power difference variables equal to the cardinality
 of $\mathcal{B}_k^n$. In contrast in (\ref{eq:smp}), 
 $s_k^n$ impacts all terms because of the coupled per-user total power constraints.
 
 Two example transformations that are valid DoV transformations, i.e., satisfying (\ref{cond:difference}) (\ref{cond:equalpower}) (\ref{cond:posdirect}), are as follows:
 \begin{enumerate}
  \item Two-tone DoV transformation: \begin{equation}\label{eq:2toneDoV}
  s_k^n = \begin{cases} t_k^n - t_{k-1}^n + P^{n,\mathrm{tot}} \gamma_k^n ~ ,k>1 \\ 
		t_1^n - t_{K}^n + P^{n,\mathrm{tot}} \gamma_1^n ~~~ ,k=1 \ \end{cases}
 \end{equation}
  \item Three-tone DoV transformation: \begin{equation}
  s_k^n = \begin{cases} - t_{k+1}^n + 2 t_k^n - t_{k-1}^n + P^{n,\mathrm{tot}} \gamma_k^n~ ,k>1 ~ \& ~ k<N \\ 
		- t_{2}^n + 2 t_1^n - t_{K}^n + P^{n,\mathrm{tot}} \gamma_1^n~~~ ,k=1 \\
		- t_{1}^n + 2 t_K^n - t_{K-1}^n + P^{n,\mathrm{tot}} \gamma_K^n~~~ ,k=N\end{cases}
 \end{equation}
 \end{enumerate}
 
 The two-tone DoV transformation has a coupling over two consecutive tones, 
 i.e., $\mathrm{card}(\mathcal{B}_k^n)=2$, whereas for the three-tone DoV
 transformation $\mathrm{card}(\mathcal{B}_k^n)=3$.
 
 \subsection{Iterative Power Difference Balancing}\label{sec:idb}
 
 Our RT-DSM algorithm design starts from the proposed reformulated problem (\ref{eq:rtsmp}).
 As the DSM problem corresponds to an NP-hard nonconvex
 problem \cite{DSMluo}, we propose an iterative coordinate ascent approach to tackle it, 
 which we refer to as \emph{iterative power difference balancing}.
 More specifically, it consists in sequentially updating/optimizing 
 one power difference variable at a time. The corresponding one-dimensional optimization
 problem is given in (\ref{eq:1Drtsmp}), where the optimization variable
 is the power difference variable $t_k^n$. 
 
 To identify the coupling level, we define a set $\mathcal{A}_k^n$ as follows,
 \begin{displaymath}
 \mathcal{A}_k^n = \left\{j \in \mathcal{K}| \beta_j^n(k) \neq 0 \right\}
 \end{displaymath}
 which denotes the set of tones for user $n$ and tone $k$ with transmit powers 
 that are influenced by power difference variable $t_k^n$.
 The objective function in (\ref{eq:1Drtsmp}) is 
 coupled over multiple tones, depending on the cardinality of $\mathcal{A}_k^n$. 
 However a proper choice of the DoV transformation results in a small coupling level, 
 reducing the sum to only a few terms.
 For instance, for the two-tone DoV transformation (\ref{eq:2toneDoV}), this corresponds
 to two terms, which means that we only consider two tones in the objective function
 and the constraints. The constraints correspond to plain bound constraints where
 the bounds $t_k^{n,\mathrm{min}}$ and $t_k^{n,\mathrm{max}}$ are simple constants
 that depend on the other power difference variables which are kept constant in the considered iteration.
 By updating the power difference variables one at a time, the total power $P^n$
 is not changed because of the zero per-user total power change property (\ref{cond:difference}).
 Each update results in an improved
 objective function value though, guaranteeing a monotonously improving performance.
 We want to highlight that, in contrast to typical existing DSM algorithms, IPDB solves the 
 problem in the \emph{primal domain} instead of the dual domain, avoiding 
 all issues related to a possible non-zero duality gap. 
 
 The one-dimensional problem (\ref{eq:1Drtsmp}) however still corresponds to a nonconvex
 problem and therefore we propose to solve it with a plain one-dimensional (1D)
 exhaustive search, where the interval $[t_k^{n,\mathrm{min}},t_k^{n,\mathrm{max}}]$
 is discretized in small steps. This can be seen as an exact line search based on a 1D
 grid search. Note that iterative grid-based exhaustive one-dimensional
 searches have been shown to be very promising in DSM literature,
 such as for the iterative spectrum balancing (ISB) algorithm \cite{ISB_WeiYu,ISB_Raphael}. We emphasize however that these existing algorithms
 focus on dual solutions where the discretization is applied to the primal
 variables, which are transmit powers. In our case, we focus on a primal solution where we consider power difference
 variables. As a result, we claim that we can make the discretization coarser, 
 because power difference variables focus on differences between tones.
 Taking into account the fact that the channels (direct as well as crosstalk
 channels) vary over tones with some degree of smoothness, the optimal transmit spectra do not differ significantly
 from one tone to the next, a property that has also been referred to as spectral correlation \cite{Tsiaflakis2012}.
 Therefore we propose to use a fine discretization for small difference 
 values and a coarse discretization for large differences.
 More specifically we choose a logarithmicly scaled discretization.
 To obtain this we define the following sets
 \begin{equation}
 \begin{array}{lcl}
 \mathcal{F} & = & \big\{ x| 10 \log_{10} (x) = -140+k \delta, k \in \mathbb{Z} \big\}\\
 \mathcal{J}^+ & = & F \cap [t_k^{n,\mathrm{min}}, t_k^{n,\mathrm{max}}]\\
 \mathcal{J} & = & \mathcal{J}^+ \cup \{0\} \cup (-\mathcal{J}^+),
 \end{array}
 \end{equation}
 where $\delta$ is a discretization variable (referred to as granularity). 
 In the case of dual algorithms
 such as ISB, $\delta=0.5$ dBm/Hz is typically chosen. However for IPDB, we
 show in Section~\ref{sec:gran} that a coarser granularity can be chosen (e.g., 1dBm/Hz,) without 
 significantly impacting the final accuracy, which then reduces computational complexity significantly.
 We note that the zero element is included in the set $\mathcal{J}$ to 
 maintain monotonicity. 
 
 The resulting grid-based search approach for the 1D problem corresponds to problem 
 (\ref{eq:1Drtsmpdiscr}), where the feasible space consists of set $\mathcal{J}$.
 \begin{equation}\label{eq:1Drtsmpdiscr}
  \maximize_{t_k^n \in \mathcal{J}} f_k^n(t_k^n)
 \end{equation}
  
 The full IPDB algorithm is given in Algorithm~\ref{algo:rtso}. In line~1, the
 power difference variables and the granularity $\delta$ are initialized. In line 2, the constants $\gamma_k^n$
 are initialized satisfying two different constraints. A straightforward
 choice here is $\gamma_k^n=1/K$, which corresponds to an
 equal power allocation over all tones, i.e., $s_k^n=P^{n,\mathrm{tot}}/K$,
 which typically satisfies all power constraints in (\ref{eq:smp}) and (\ref{eq:rtsmp}).
 The repeat loop in line 3 is a loop that can go until some stopping criterion is
 achieved or until some real-time deadline is reached.
 The loop in line~4 is the per-user loop. Note that the user order is not defined
 and can be arbitrarily chosen. In fact this user order can also have multiple instances
 of the same user. Line 5 is the inner per-user iteration with a maximum of
 $I$ iterations. Line 6 is the per-tone loop. Again, the tone
 order is not necessarily consecutive but can be arbitrary. Line~7 is the
 only line that involves an update of the transmit powers and corresponds to a 
 one-dimensional power difference variable update by a 1D exhaustive grid-based search
 of problem (\ref{eq:1Drtsmpdiscr}) over the values $\mathcal{J}$.
 With the DoV transformation (\ref{eq:transformation})
 the corresponding updated transmit powers can be obtained.  In line 9
 the power difference variables are then reset and the constants $\gamma_k^n$ are updated
 so as to keep the transmit powers fixed. Although the latter two
 actions are not necessary from a theoretical point of view, they are seen to improve
 the performance, as the values around $t_k^n=0$ are discretized at a finer
 granularity. This can be seen as a \emph{recentering} operation
 so as to fully benefit from the logarithmicly scaled discretization.
 As a result a fine granularity in transmit powers can be obtained through
 a sum of coarse power difference variables (that are not coarse everywhere). 
 This results in a good final solution accuracy as demonstrated in 
 Section~\ref{sec:simulations}. Lines~12 and 13 correspond to an inequality procedure to consider
 inequalities and an equalization procedure, as discussed in 
 Section~\ref{sec:equalitymasks} and \ref{sec:equalization}, respectively.
 Note however that these steps are not necessary and can be disregarded at this point.
  
 We now analyze different aspects and properties of the IPDB algorithm:
 
 \subsubsection{Tunability}
 We want to highlight that the IPDB algorithm offers flexibility 
 in choosing the user order, the tone order, the number of inner iterations,
 the granularity $\delta$, and in the initialization of the parameters $\gamma_k^n$. Different such choices
 are evaluated in Section~\ref{sec:configurations}. 
 
 \subsubsection{Real-time Property}
 The IPDB algorithm satisfies the real-time property from Definition~\ref{def:rtso}: it can be stopped 
 at any moment as it satisfies the constraints after every single update
 of the power difference variables, which have a one-to-one mapping to the
 transmit powers through (\ref{eq:transformation}). The concrete improved real-time behaviour is
 demonstrated in Section~\ref{sec:realtimeproperty}.
 
  \begin{algorithm}
    \caption{Iterative Power Difference Balancing}\label{algo:rtso}
    \begin{algorithmic}[1]
    \State Initialize $\delta, t_k^n \gets 0, \forall n, \forall k$
    \State Initialize $\gamma_k^n, \forall n, \forall k$ satisfying (\ref{cond:equalpower}) and $0\leq \gamma_k^n \leq \frac{s_k^{n,\mathrm{mask}}}{P^{n,\mathrm{tot}}}$
    \Repeat
    \For{ $n \gets$ userOrder}  
     \For { $i \gets 1, I$}     
      \For{ $k \gets$ toneOrder} 
       \State $t_k^n \gets$ Solve (\ref{eq:1Drtsmpdiscr}); $s_k^n \gets (\ref{eq:transformation}), \forall k \in \mathcal{A}_k^n$
      \EndFor
      \State $t_k^n \gets 0, \gamma_k^n \gets s_k^n/P^{n,\mathrm{tot}}, \forall k$
     \EndFor
    \EndFor
    \State Inequality procedure Algorithm~\ref{algo:rtso_inequal}
    \State Equalization procedure Algorithm~\ref{algo:rtso_equal}
    \Until{convergence stop criterion}
    \end{algorithmic}
    \end{algorithm}
    
  \subsubsection{Complexity}
  The computational complexity analysis of IPDB is rather straightforward. 
  Most of the complexity results from line 7, which corresponds to a simple 1D exhaustive grid-based search. 
  Under a given computational and compute power budget it is easy to determine 
  the number of updates that can be performed, which demonstrates the benefit of the real-time property of IPDB.
  
  \subsubsection{Monotonicity}
  Each update results in a non-decreasing feasible objective function value. As 
  a result IPDB has an interesting monotonicity and scalability property where
  more computation time or compute power consistently results in a better obtained solution.
  
  \subsubsection{Convergence}
  As the IPDB algorithm is a coordinate search method, the convergence behaviour 
  is inherited from such methods. The looser the coupling between the 
  coordinate ascent variables, the faster the convergence \cite{nocedal}.
  However, it is important to highlight that because of the real-time property
  that ensures constraint satisfaction after each single update, and the
  monotonicity property that ensures a non-decreasing objective function value,
  it is not extremely important that full convergence is reached when performing the IPDB algorithm.
  Fast numerical convergence results (up to 99\% and 99.9\% of full performance convergence) 
  are demonstrated in Section~\ref{sec:configurations}.\\
  
  Finally, we want to highlight that although we employ a DoV transformation
  with differences between the transmit powers on different tones for one user, one could 
  in principle also employ differences between the transmit powers of different users on a single tone or different
  tones if there are per-tone sum power constraints or total network sum power constraints, respectively.
    
  \subsection{Inequality constraints}\label{sec:equalitymasks}
  
  In this section we consider inequality constraints for the per-user total power
  constraints as given by the following DSM problem
  
  \begin{equation}\label{eq:smpinequality}
    \begin{array}{cl}
    {\displaystyle \maximize_{s_k^n, k \in \mathcal{K}, n \in \mathcal{N}}} & 
    {\displaystyle \sum_{n \in \mathcal{N}} w_n R^n(\mathbf{s}_k, k \in \mathcal{K})}\\
    \mathrm{subject~to} & {\displaystyle P^n(s_k^n, k \in \mathcal{K}) \leq P^{n,\mathrm{tot}}}, \forall n \in \mathcal{N}\\
    & 0 \leq s_k^n \leq s_k^{n,\mathrm{mask}}, \forall n \in \mathcal{N}, \forall k \in \mathcal{K}.\\
    \end{array} 
  \end{equation}
  
  To extend IPDB to also cover inequality constraints, we propose
  an inequality procedure that allows to reduce the per-user total powers below
  $P^{n,\mathrm{tot}}$, whenever this improves the weighted sum of achievable data rates.
  This procedure is given in Algorithm~{\ref{algo:rtso_inequal}}. 
  We have $\alpha > 1$ and $0 < \beta < 1$. $s_\alpha$ computes a value larger
  than the current value $s_k^n$ while satisfying the per-user total power constraint
  as well as the spectral mask constraint. $s_\beta$ computes a value smaller than 
  the current value $s_k^n$. In line 4, a per-tone weighted sum of bit rates evaluation 
  for user $n$ is performed to check if the weighted sum of bit rates
  can be increased by increasing or decreasing the transmit power $s_k^n$.
  We note that $\mathbf{s}_k|_{s_k^n=\hat{s}_k^n}$ equals $\mathbf{s}_k$ with $s_k^n$ being replaced by $\hat{s}_k^n$.
  This is repeated for all tones. This inequality procedure does not violate the real-time property as 
  the constraints of (\ref{eq:smpinequality}) are satisfied after every single update.
  Also the monotonicity property is not violated as each update results in 
  a non-decreasing feasible objective function value.
  
    \begin{algorithm}
    \caption{Inequality procedure user $n$}\label{algo:rtso_inequal}
    \begin{algorithmic}[1]
    %\State Initialize $t_k^n \gets 0, \forall n, \forall k$
    %\State Initialize $\gamma_k^n, \forall n, \forall k$ satisfying (\ref{cond:equalpower}) and $0\leq \gamma_k^n \leq \frac{s_k^{n,\mathrm{mask}}}{P^{n,\mathrm{tot}}}$
    %\Repeat
    %\For{ $n \gets$ userOrder}  
    % \For { $i \gets 1, I$}     
     \For{ $k \gets$ toneOrder} 
      \State ${\displaystyle s_\alpha \gets \min(\alpha s_k^n,P^{n,\mathrm{tot}}-\sum_{q \in \mathcal{K}\backslash k} s_q^n, s_k^{n,\mathrm{mask}})}$
      \State $s_\beta \gets \beta s_k^n$
      \State $s_k^n \gets {\displaystyle \argmax_{\hat{s}_k^n \in \{s_\alpha, s_k^n, s_\beta \}} \sum_{m \in \mathcal{N}} w_m b_k^m(\mathbf{s}_k|_{s_k^n=\hat{s}_k^n})}$
     \EndFor
    %  \State $s_k^n \gets (\ref{eq:transformation}), \forall k$, $t_k^n \gets 0, \forall k$
    % \EndFor
    %\EndFor
    %\Until{convergence stop criterion}
    \end{algorithmic}
    \end{algorithm}
    
 \subsection{Randomization and Equalization}\label{sec:equalization}
 
 As mentioned in Section~\ref{sec:idb}, the IPDB algorithm is tunable in
 the tone order, the user order, and in the initial choice of $\gamma_k^n$, where the latter has a one-to-one mapping with
 the initial transmit powers if the power difference variables $t_k^n$ are fixed.
 We can use randomized values, i.e., a random tone order, random user order 
 and random initial transmit powers. Randomization
 in iteration orders and initial conditions has been shown to be effective in
 several cases in literature \cite{Nesterov2013}. It is shown in 
 Section~\ref{sec:simulations} that randomization indeed results
 in performance gains. 
 
 However, randomization also has some side effects 
 which may not always be desirable. For instance,
 when randomizing the initial transmit powers, the
 resulting transmit spectra may have a very non-smooth behaviour in the sense that
 transmit powers differ significantly from one tone to the next. For instance, 
 in Figure~\ref{fig:isbRndSpectra} the resulting transmit spectra are
 shown when applying IPDB to a 2-user scenario (corresponding to the blue and
 the green curve) with randomized initial transmit spectra. The transmit
 spectra between tones 45 and 115 display significant jumps. 
 The reason is that by starting from very different initial transmit powers on consecutive tones 
 IPDB can converge to very different per-tone solutions in consecutive tones. We want to highlight that the non-smooth behaviour results in
 a typically better overall performance and prevents convergence to
 a very poor solution for which a poor choice is systematically made
 over multiple consecutive tones.
 
 \begin{figure}[!t]
  \centering
  \includegraphics[width=0.95\columnwidth]{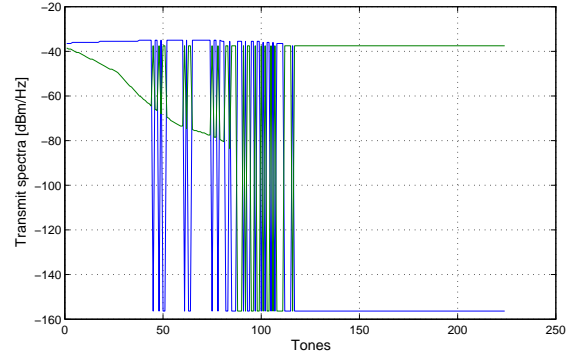}
  \caption{Resulting transmit powers for 2-user ADSL scenario of Figure~\ref{fig:2adsl} when using IPDB with randomized initial transmit spectra and without equalization. The transmit spectra
  between tones 45 and 115 display significant jumps.}
  \label{fig:isbRndSpectra}
 \end{figure}
 
 However, such non-smooth solution behaviour is not always desirable in practice
 and therefore we propose an equalization procedure that can smooth it out. 
 This procedure is given in Algorithm~\ref{algo:rtso_equal}
 and consists of one simple \emph{for} loop over the tones. In each loop, 
 it is checked if a spike can be detected between three (almost) consecutive tones $k$, $k+1$ and $k+3$. If a down spike is detected (with more than 10~dB difference), 
 this is averaged out. If an up spike is detected (with more than 10~dB difference) this is flattened out.
 The reason $k+3$ instead of $k+2$ is considered for the third tone is that 
 this way one-tone wide spikes as well as two-tone wide spikes can be detected.
 Although the above procedure is a very low-complexity operation, it has very good
 equalization performance with smooth resulting transmit spectra, as 
 demonstrated in Section~\ref{sec:simulations}.
 From an optimization point of view, the equalization procedure can be seen as
 a procedure that allows to obtain overall solutions where similar per-tone solutions 
 are chosen in neighboring tones.
 
 We want to highlight that this equalization procedure does not
 have to be called in all iterations but only after every so many outer iterations.
 
 The equalization procedure can slightly violate the real-time property
 during execution of line 10 and lines 12 to 14, but we want to highlight
 that these steps are of much smaller granularity than the main computational
 step of IPDB, i.e. line 7 of Algorithm~\ref{algo:rtso}.
 
 Although the monotonicity 
 property may be violated whenever the equalization procedure is called, overall
 this seems to provide a better performance
 as demonstrated in Section~\ref{sec:simulations}.
 
 \begin{algorithm}
    \caption{Equalization procedure user $n$}\label{algo:rtso_equal}
    \begin{algorithmic}[1]
    %\State Initialize $t_k^n \gets 0, \forall n, \forall k$
    %\State Initialize $\gamma_k^n, \forall n, \forall k$ satisfying (\ref{cond:equalpower}) and $0\leq \gamma_k^n \leq \frac{s_k^{n,\mathrm{mask}}}{P^{n,\mathrm{tot}}}$
    %\Repeat
    %\For{ $n \gets$ userOrder}  
    % \For { $i \gets 1, I$}    
    \State $p \gets \sum_{k \in \mathcal{K}} s_k^n$
     \For{ $k \gets 1\ldots K-3$}
      \State $u_l \gets 10 \log_{10}(s_{k}^n)-10$
      \State $u_r \gets 10 \log_{10}(s_{k+3}^n)-10$
      \State $s_{\mathrm{db}} \gets 10 \log_{10}(s_{k+1}^n)$
      \State $d_l \gets 10 \log_{10}(s_{k}^n)+10$
      \State $d_r \gets 10 \log_{10}(s_{k+3}^n)+10$
      \If{$s_{\mathrm{db}} < u_l$ and $s_{\mathrm{db}} < u_r$}
	\State $s_m \gets (s_k^n + s_{k+1}^n + s_{k+3}^n)/3$
	\State $s_k^n \gets s_m, s_{k+1}^n \gets s_m, s_{k+3}^n \gets s_m$
      \ElsIf{ $s_{\mathrm{db}} > d_l$ and $s_{\mathrm{db}} > d_r$}
 	\State $s_{k+1}^n \gets \min(s_k^n,s_{k+3}^n)$
 	\State $p_r \gets \sum_{k \in \mathcal{K}} s_k^n$
 	\State $s_k^n \gets s_k^n ~ p / p_r, k \in \mathcal{K}$
      \EndIf
     \EndFor
    %  \State $s_k^n \gets (\ref{eq:transformation}), \forall k$, $t_k^n \gets 0, \forall k$
    % \EndFor
    %\EndFor
    %\Until{convergence stop criterion}
    \end{algorithmic}
    \end{algorithm}

\section{Simulations and Performance Analysis}\label{sec:simulations}

In this section the performance of the IPDB algorithm is evaluated for 
different settings and for different performance
metrics. The performance will be compared with that of the popular ISB
algorithm which is also a coordinate ascent grid-based search algorithm. In contrast to IPDB which 
operates in the primal domain, ISB is based on a combination of a dual decomposition
approach with a discrete per-tone coordinate ascent grid-based search. ISB
does not have the real-time property and is thus not an RT-DSM algorithm.
Wireline DSL as well as wireless LTE settings will be considered. 

For our wireline DSL simulations in Sections~\ref{sec:configurations} \ref{sec:gran} \ref{sec:inequalsim}, 
we use realistic DSL simulators, which have been validated in practice and 
are aligned with standards. We consider 24AWG twisted copper pair
lines. The maximum transmit power is 20.4~dBm for the ADSL and ADSL2+ scenarios, and 11.5~dBm
for the VDSL scenarios. The SNR gap is chosen at 12.9~dB for the DSL scenarios, corresponding
to a coding gain of 3~dB, a noise margin of 6~dB, and a target symbol error 
probability of $10^{-7}$. The tone spacing is 4.3125~kHz. The DMT symbol rate is 4~kHz. The weights
$w_n$ are chosen equal for all users ($n=1\ldots N$), namely $w_n=1/N$, unless specified otherwise.
When the equalization procedure is activated in Algorithm~\ref{algo:rtso}, it
is only performed after each fifth outer iteration. 

A concrete wireless LTE heterogeneous network setting
will be discussed in Section~\ref{sec:lte}.

\subsection{IPDB performance: ADSL Case}\label{sec:configurations}

The ADSL scenario under consideration is the
near-far scenario shown in Figure~\ref{fig:2adsl}, i.e., a 2-user
downstream scenario with one far-user connected to the central office (CO) and
a second near-user connected to a remote terminal (RT). This near-far type of scenarios is
quite common in practice and has received a lot of attention in DSL literature. 
The underlying optimization problem is known to display a very nonconvex behaviour 
for which locally optimal DSM methods can perform poorly \cite{dsb}.
In our simulations, the RT-connected user is given a weight of 0.1 and 
the CO-connected user a weight of 0.9, to prevent that the latter is being
allocated a too small achievable data rate.

 \begin{figure}[!t]
  \centering
  \includegraphics[width=0.675\columnwidth]{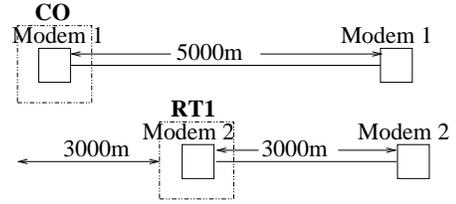}
  \caption{Near-far 2-user ADSL downstream scenario with central office (CO) and remote terminal (RT)}
  \label{fig:2adsl}
 \end{figure}

The performance of both IPDB and ISB are evaluated
for different configurations exploiting the tunability of IPDB. 
The tone order corresponds to one configuration setting for which we test four choices. 
TO 1 and TO 2 correspond to tone orders $[1:K]$ and $[K:-1:1]$, respectively. TO 3 selects
TO 1 or TO 2 with probability of 50\% each, each time line 6 of IPDB is entered. 
TO 4 corresponds to a tone order which is a fully random permutation 
of the tone set $[1:K]$, and which changes each time line 6 of IPDB is entered.

A second configuration setting corresponds to the initial transmit power spectra. For this
we consider two choices: (1) EP corresponds to an equal power allocation, i.e., 
$\gamma_k^n=1/K~\forall k,n$, (2) RP corresponds to a random power allocation
with uniformly distributed probabilities in dB scale while satisfying
(\ref{cond:equalpower}) and $0\leq \gamma_k^n \leq s_k^{n,\mathrm{mask}}/P^{n,\mathrm{tot}}$.
The number of inner iterations $I$ is fixed at 1.

We consider IPDB with and without equalization (Algorithm~\ref{algo:rtso_equal}),
corresponding to EQ ON and EQ OFF, respectively.
Three different DoV transformations are evaluated, which correspond to
(\ref{eq:2toneDoV}), (\ref{eq:2toneDoVrand}) and (\ref{eq:3tone2}), respectively, and
where $\pi$ is a vector that corresponds to a random permutation of vector
$[1,\ldots,K]$:

\begin{enumerate}
  \item Two-tone rand DoV transformation: \begin{equation}\label{eq:2toneDoVrand}
  s_k^n = t_k^n - t^n_{\pi(k)} + P^{n,\mathrm{tot}} \gamma_k^n
 \end{equation}
  \item Three-tone 2 DoV transformation: \begin{equation}\label{eq:3tone2}
  s_k^n = \begin{cases} 2 t_k^n - t_{k+1}^n - t_{k+2}^n + P^{n,\mathrm{tot}} \gamma_k^n~~~ ,k<N-2 \\ 
		2 t_{N-1}^n - t_{N}^n - t_{1}^n  + P^{n,\mathrm{tot}} \gamma_{N-1}^n~~~~ ,k=N-1 \\
		2 t_{N}^n - t_{1}^n - t_{2}^n  + P^{n,\mathrm{tot}} \gamma_N^n~~~~~~~~ ,k=N\end{cases}
 \end{equation}
 \end{enumerate}

We also consider different granularities for the discrete searches,
where we consider the standard 0.5~dBm/Hz \cite{OSB, ISB_Raphael} for the transmit powers $s_k^n$ in ISB, and
a coarser granularity of 1~dBm/Hz and 10~dBm/Hz for the power difference variables $t_k^n$ in IPDB.

Finally, for all configurations, the results are averaged over 15 different runs
to obtain an averaged performance.

\subsubsection{Weighted Sum of Achievable Data Rates Performance}

In Table~\ref{tab:datarate} the weighted sum of achievable data rates performance is compared for IPDB
and ISB under the different configurations.

The standard ISB configuration corresponds to the settings TO 1, EQ OFF, EP
and 0.5~dBm/Hz granularity. For this standard setting, ISB has a weighted
sum of achievable data rates performance of 1.5549~Mbps. This increases spectacularly
when applying randomized initial transmit powers, i.e., RP with up to 1.8274~Mbps. In addition,
applying the equalization procedure (EQ ON) further improves the performance
up to 1.8799~Mbps, when combined with RP. This is quite surprising as the equalization procedure not only provides 
smoother resulting transmit spectra (as shown in Section~\ref{sec:equalperf}) but in addition 
also improves the achievable data rate performance. The explanation is that in the case of
randomization most of the tones converge to good per-tone solutions.
When combining this with the equalization procedure, this will cause the
fewer poor per-tone solutions to be forced to the larger set of good 
per-tone solutions, resulting in a better overall performance.

The performance of IPDB is good when using the two-tone rand DoV transformation
of (\ref{eq:2toneDoVrand}). For a 1~dB granularity, one can see a performance
improvement of up to 22\% compared to the standard ISB configuration, and up to 1\%
compared to the best ISB configuration with RP and EQ ON. The IPDB performance is not so good 
for the two-tone and three-tone 2 DoV transformations of (\ref{eq:2toneDoV}) 
and (\ref{eq:3tone2}), respectively. However when
combining these DoV transformations with the equalization procedure,
their performance is improved spectacularly. Reducing the granularity from 
1~dBm/Hz steps to 10~dBm/Hz steps decreases the performance, 
for all cases. It is interesting to notice that the tone order does not play a significant role.

In terms of achievable data rate performance, it can be summarized that the two-tone 
rand DoV transformation with equalization activated and 1~dBm/Hz granularity offers
the best performance, and performs better than all ISB configurations.

\begin{table*}
 \renewcommand{\arraystretch}{1.3}
 \caption{Weighted sum of data rates performance [Mbps] for the near-far CO-RT scenario of Fig.~\ref{fig:2adsl} with different settings for IPDB. Method Two-tone rand uses DoV transformation (\ref{eq:2toneDoVrand}).
  Tone order (TO) 1 = [1:K], tone order 2 = [K:-1:1], tone order 3 = uniform probability of 1 and 2, tone order 4 = random permuation,
 EP = constant equal init power, RP = random init power satisfying total power constraint. Two right most columns correspond to ISB performance. }\label{tab:datarate}
\begin{tabular}{|c|c|c|c|c|c|c|c|c|c|c|c|c|c|c|c|}
 \hline
 \multicolumn{2}{|c|}{Method} & \multicolumn{4}{|c|}{Two-tone} & \multicolumn{4}{|c|}{Two-tone rand} & \multicolumn{4}{|c|}{Three-tone 2} & \multicolumn{2}{|c|}{ISB}\\
 \hline
 \multicolumn{2}{|c|}{Equalization} & \multicolumn{2}{|c|}{EQ OFF} & \multicolumn{2}{|c|}{EQ ON}& \multicolumn{2}{|c|}{EQ OFF} & \multicolumn{2}{|c|}{EQ ON}& \multicolumn{2}{|c|}{EQ OFF} & \multicolumn{2}{|c|}{EQ ON}& EQ OFF & EQ ON\\
 \hline
 \multicolumn{2}{|c|}{Granularity} & 1dB & 10dB & 1dB & 10dB & 1dB & 10dB & 1dB & 10dB & 1dB & 10dB & 1dB & 10dB & 0.5dB & 0.5dB\\
 \hline
 \multirow{2}{*}{TO 1} & EP & \cellcolor{blue!25}1.5080 & 1.4223 & \cellcolor{blue!25}1.8608 & 1.6542 & \cellcolor{blue!25}1.8542 & 1.8320 & \cellcolor{blue!25}1.8978 & 1.8519& \cellcolor{blue!25}1.5143 & 1.3780 & \cellcolor{blue!25} 1.8324 & 1.3875 & \cellcolor{LightCyan}1.5549 &  \cellcolor{LightCyan}1.5549 \\ 
 \cline{2-16}
 & RP & \cellcolor{blue!25}1.3858 & 1.3739 & \cellcolor{blue!25}1.7538 & 1.7892 & \cellcolor{blue!25}1.8087 & 1.8247 & \cellcolor{blue!25}1.8951 & 1.8775  & \cellcolor{blue!25}1.4065 & 1.3930 & \cellcolor{blue!25}1.8135 & 1.7490& \cellcolor{LightCyan}1.8274 & \cellcolor{LightCyan}1.8799 \\
 \hline
 \multirow{2}{*}{TO 2} & EP & \cellcolor{blue!25}1.5826 & 1.4626 & \cellcolor{blue!25}1.7914 & 1.6759 & \cellcolor{blue!25}1.7872 & 1.8358 & \cellcolor{blue!25}1.8761 & 1.8374 & \cellcolor{blue!25}1.5066 & 1.3552 & \cellcolor{blue!25}1.8439 & 1.3298& \cellcolor{Gray}& \cellcolor{Gray} \\ 
 \cline{2-16}
 & RP & \cellcolor{blue!25}1.3724 & 1.3388 & \cellcolor{blue!25}1.7069 & 1.7689 & \cellcolor{blue!25}1.7940 & 1.8212 & \cellcolor{blue!25}1.8806 & 1.8507 & \cellcolor{blue!25}1.4114 & 1.3657 & \cellcolor{blue!25}1.7985 & 1.7705& \cellcolor{Gray} & \cellcolor{Gray}\\
 \hline
 \multirow{2}{*}{TO 3} & EP & \cellcolor{blue!25}1.5393 & 1.4459 & \cellcolor{blue!25}1.8647 & 1.6432 & \cellcolor{blue!25}1.8201 & 1.8345 & \cellcolor{blue!25}1.8731 & 1.8611 & \cellcolor{blue!25}1.4978 & 1.3772 & \cellcolor{blue!25}1.8316 & 1.3758 & \cellcolor{Gray}& \cellcolor{Gray}\\ 
 \cline{2-16}
 & RP & \cellcolor{blue!25}1.3957 & 1.3834 & \cellcolor{blue!25}1.7759 & 1.7380 & \cellcolor{blue!25}1.7975 & 1.8236 & \cellcolor{blue!25}1.8831 & 1.8503 & \cellcolor{blue!25}1.3887 & 1.4134 & \cellcolor{blue!25}1.8428 & 1.7149 & \cellcolor{Gray}&\cellcolor{Gray} \\
 \hline
 \multirow{2}{*}{TO 4} & EP & \cellcolor{blue!25}1.4428 & 1.4422 & \cellcolor{blue!25}1.8543 & 1.5426 & \cellcolor{blue!25}1.8240 & 1.8308 & \cellcolor{blue!25}1.8739 & 1.8394 & \cellcolor{blue!25}1.4564 & 1.3624 & \cellcolor{blue!25}1.7925 & 1.4202& \cellcolor{Gray}&\cellcolor{Gray} \\ 
 \cline{2-16}
 & RP & \cellcolor{blue!25}1.3972 & 1.3948 & \cellcolor{blue!25}1.8045 & 1.7071 & \cellcolor{blue!25}1.8099 & 1.8147 & \cellcolor{blue!25}1.8854 & 1.8620 & \cellcolor{blue!25}1.4265 & 1.3709 & \cellcolor{blue!25}1.8375 & 1.7113 & \cellcolor{Gray}&\cellcolor{Gray} \\
  \hline
\end{tabular}
\end{table*}

\subsubsection{Convergence Speed}

Tables~\ref{tab:convspeed99} and \ref{tab:convspeed999} display the convergence speed
in terms of the number of outer iterations (loop corresponding to lines 3-14) to converge to
99\% and 99.9\% weighted achievable sum data rate performance, respectively.

The number of outer iterations varies strongly for different IPDB configurations.
A very fast convergence is achieved for the two-tone rand DoV transformation. 
In particular, we see convergence after less than 10 outer iterations (for 99\% performance)
and 16 outer iterations (for 99.9\% performance), when using the settings 
EQ OFF and 1~dBm/Hz granularity. For the EQ ON setting, this increases to only 
20 outer iterations for both 1~dBm/Hz and 10~dBm/Hz granularity. 
The other DoV transformations require up to 200 iterations to converge. 
In general, a larger granularity results in a slower convergence. 

ISB only requires 2 to 6 outer iterations to converge. 
However, we want to highlight that ISB requires a dual 
optimization step to satisfy the total power constraints within each outer 
iteration, which is not required for IPDB. The complexity per outer iteration
is thus larger for ISB.

For illustration, in Figures~\ref{fig:wrateevolution} and \ref{fig:powerevolution} 
the evolution of the weighted sum of achievable data rates and the per-user total powers 
is displayed as a function of the outer iterations for a single run. It 
can be seen that 99\% and 99.9\% performance are achieved after
8 and 12 outer iterations. The per-user total powers constraints are always satisfied, which
demonstrates the real-time property of IPDB.

 \begin{figure}[!t]
  \centering
  \includegraphics[width=0.95\columnwidth]{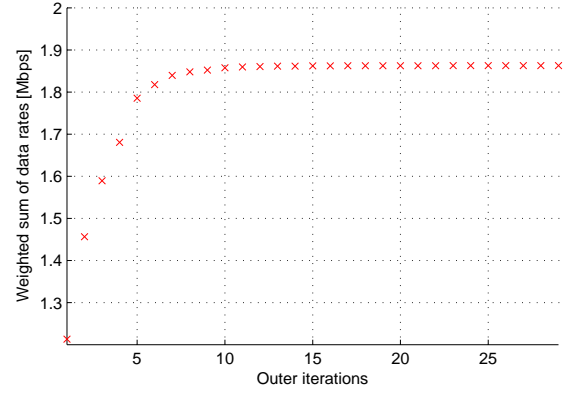}
  \caption{Evolution of weighted sum of data rates for IPDB algorithm}
  \label{fig:wrateevolution}
 \end{figure}

\begin{figure}[!t]
  \centering
  \includegraphics[width=0.95\columnwidth]{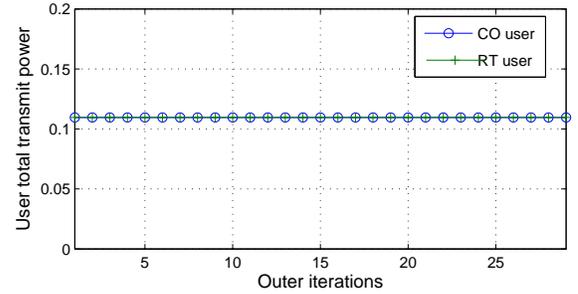}
  \caption{Evolution of user powers for IPDB algorithm}
  \label{fig:powerevolution}
 \end{figure}

\begin{table*}
 \renewcommand{\arraystretch}{1.3}
 \caption{Convergence speed [nb of outer iterations] to converge to 99\% performance for near-far CO-RT scenario of Fig.~\ref{fig:2adsl} with different settings for IPDB. Method Two-tone rand is with permutated tones on the transformation.
 Tone order (TO) 1 = [1:K], tone order 2 = [K:-1:1], tone order 3 = uniform probability of 1 and 2, tone order 4 = random permuation,
 EP = constant equal init power, RP = random init power satisfying total power constraint. Two right most columns correspond to ISB performance.}\label{tab:convspeed99}
\begin{tabular}{|c|c|c|c|c|c|c|c|c|c|c|c|c|c|c|c|}
 \hline
 \multicolumn{2}{|c|}{Method} & \multicolumn{4}{|c|}{Two-tone} & \multicolumn{4}{|c|}{Two-tone rand} & \multicolumn{4}{|c|}{Three-tone 2} & \multicolumn{2}{|c|}{ISB}\\
 \hline
 \multicolumn{2}{|c|}{Equalization} & \multicolumn{2}{|c|}{EQ OFF} & \multicolumn{2}{|c|}{EQ ON}& \multicolumn{2}{|c|}{EQ OFF} & \multicolumn{2}{|c|}{EQ ON}& \multicolumn{2}{|c|}{EQ OFF} & \multicolumn{2}{|c|}{EQ ON}& EQ OFF & EQ ON\\
 \hline
 \multicolumn{2}{|c|}{Granularity} & 1~dB & 10~dB & 1~dB & 10~dB & 1~dB & 10~dB & 1~dB & 10~dB & 1~dB & 10~dB & 1~dB & 10~dB & 0.5~dB & 0.5~dB\\
 \hline
 \multirow{2}{*}{TO 1} & EP & \cellcolor{blue!25}46.00 & 40.00 & \cellcolor{blue!25}150.00 & 162.00 & \cellcolor{blue!25}5.67 & 8.20 & \cellcolor{blue!25}10.53 & 19.67 & \cellcolor{blue!25}80.00 & 56.00 & \cellcolor{blue!25}101.00 & 103.00 & \cellcolor{LightCyan} 2.00 &  \cellcolor{LightCyan} 2.00\\ 
 \cline{2-16}
 & RP & \cellcolor{blue!25}41.33 & 61.67 & \cellcolor{blue!25}124.13 & 162.13 & \cellcolor{blue!25}8.60 & 14.60 & \cellcolor{blue!25}17.87  & 15.00 & \cellcolor{blue!25}31.67 & 44.13 & \cellcolor{blue!25}133.20 & 169.00 &  \cellcolor{LightCyan}2.00 & \cellcolor{LightCyan}5.00 \\
 \hline
  \multirow{2}{*}{TO 2} & EP &\cellcolor{blue!25}29.00 & 52.00 & \cellcolor{blue!25}61.00 & 174.00 & \cellcolor{blue!25}5.20 & 8.00 & \cellcolor{blue!25}10.87 & 8.73 & \cellcolor{blue!25}22.00 & 59.00 & \cellcolor{blue!25}116.00 & 177.00 & \cellcolor{Gray}& \cellcolor{Gray} \\ 
 \cline{2-16}
  & RP & \cellcolor{blue!25}47.73 & 57.40 & \cellcolor{blue!25}99.60 & 143.87 & \cellcolor{blue!25}8.33 & 13.40 & \cellcolor{blue!25}15.53 & 20.20 & \cellcolor{blue!25}30.47 & 45.27 & \cellcolor{blue!25}122.07 & 163.93& \cellcolor{Gray} & \cellcolor{Gray}\\
 \hline
  \multirow{2}{*}{TO 3} & EP & \cellcolor{blue!25}30.20 & 43.93 & \cellcolor{blue!25}104.80 & 169.53 & \cellcolor{blue!25}5.33 & 8.13 & \cellcolor{blue!25}10.60 & 11.67 & \cellcolor{blue!25}28.60 & 88.40 & \cellcolor{blue!25}111.20 & 144.87 & \cellcolor{Gray}& \cellcolor{Gray}\\ 
 \cline{2-16}
  & RP & \cellcolor{blue!25}28.53 & 44.00 & \cellcolor{blue!25}110.27 & 148.13 & \cellcolor{blue!25}9.47 & 14.47 & \cellcolor{blue!25}11.80 & 13.80 & \cellcolor{blue!25}29.80 & 45.40 & \cellcolor{blue!25}125.40 & 171.33& \cellcolor{Gray}&\cellcolor{Gray} \\
 \hline
 \multirow{2}{*}{TO 4} & EP & \cellcolor{blue!25}16.73 & 43.33 & \cellcolor{blue!25}42.87 & 162.87 & \cellcolor{blue!25}5.27 & 8.20 & \cellcolor{blue!25}9.60 & 11.60 & \cellcolor{blue!25}36.07 & 62.40 & \cellcolor{blue!25}126.93 & 137.13& \cellcolor{Gray}&\cellcolor{Gray} \\ 
 \cline{2-16}
 & RP & \cellcolor{blue!25}43.73 & 60.53 & \cellcolor{blue!25}94.27 & 162.07 & \cellcolor{blue!25}8.40 & 13.80 & \cellcolor{blue!25}13.33 & 20.13 & \cellcolor{blue!25}26.33 & 41.87 & \cellcolor{blue!25}125.80 & 163.13 & \cellcolor{Gray}&\cellcolor{Gray} \\
  \hline
\end{tabular}
\end{table*}

\begin{table*}
 \renewcommand{\arraystretch}{1.3}
 \caption{Convergence speed [nb of outer iterations] to converge to 99.9\% performance for near-far CO-RT scenario of Fig.~\ref{fig:2adsl} with different settings for IPDB. Method Two-tone rand is with permutated tones on the transformation.
 Tone order (TO) 1 = [1:K], tone order 2 = [K:-1:1], tone order 3 = uniform probability of 1 and 2, tone order 4 = random permuation,
 EP = constant equal init power, RP = random init power satisfying total power constraint. Two right most columns correspond to ISB performance.}\label{tab:convspeed999}
\begin{tabular}{|c|c|c|c|c|c|c|c|c|c|c|c|c|c|c|c|}
 \hline
 \multicolumn{2}{|c|}{Method} & \multicolumn{4}{|c|}{Two-tone} & \multicolumn{4}{|c|}{Two-tone rand} & \multicolumn{4}{|c|}{Three-tone 2} & \multicolumn{2}{|c|}{ISB}\\
 \hline
 \multicolumn{2}{|c|}{Equalization} & \multicolumn{2}{|c|}{EQ OFF} & \multicolumn{2}{|c|}{EQ ON}& \multicolumn{2}{|c|}{EQ OFF} & \multicolumn{2}{|c|}{EQ ON}& \multicolumn{2}{|c|}{EQ OFF} & \multicolumn{2}{|c|}{EQ ON}& EQ OFF & EQ ON\\
 \hline
 \multicolumn{2}{|c|}{Granularity} & 1~dB & 10~dB & 1~dB & 10~dB & 1~dB & 10~dB & 1~dB & 10~dB & 1~dB & 10~dB & 1~dB & 10~dB & 0.5~dB & 0.5~dB\\
 \hline
  \multirow{2}{*}{TO 1} & EP & \cellcolor{blue!25}84.00 & 68.00 & \cellcolor{blue!25}155.00 & 174.00 & \cellcolor{blue!25}11.07 & 15.53 & \cellcolor{blue!25}14.00 & 25.13 & \cellcolor{blue!25}83.00 & 89.00 & \cellcolor{blue!25}169.00 & 104.00 & \cellcolor{LightCyan} 2.00&  \cellcolor{LightCyan}2.00 \\ 
 \cline{2-16}
  & RP & \cellcolor{blue!25}118.33 & 156.40 & \cellcolor{blue!25}144.60 & 191.53 & \cellcolor{blue!25}14.07 & 22.13 & \cellcolor{blue!25}22.60 & 22.20 & \cellcolor{blue!25}93.53 & 138.07 & \cellcolor{blue!25}157.07 & 189.80 &  \cellcolor{LightCyan}2.25 & \cellcolor{LightCyan}6.25 \\
 \hline
  \multirow{2}{*}{TO 2} & EP & \cellcolor{blue!25}32.00 & 66.00 & \cellcolor{blue!25}63.00 & 200.00 & \cellcolor{blue!25}10.60 & 15.13 & \cellcolor{blue!25}15.87 & 13.87 & \cellcolor{blue!25}50.00 & 70.00 & \cellcolor{blue!25}184.00 & 198.00& \cellcolor{Gray}& \cellcolor{Gray} \\ 
 \cline{2-16}
  & RP &\cellcolor{blue!25}125.53 & 137.53 & \cellcolor{blue!25}124.40 & 168.13 & \cellcolor{blue!25}13.33 & 20.73 & \cellcolor{blue!25}21.13 & 27.07 & \cellcolor{blue!25}82.47 & 127.80 & \cellcolor{blue!25}163.27 & 189.87 & \cellcolor{Gray} & \cellcolor{Gray}\\
 \hline
  \multirow{2}{*}{TO 3} & EP & \cellcolor{blue!25}45.60 & 58.33 & \cellcolor{blue!25}147.47 & 188.40 & \cellcolor{blue!25}10.80 & 15.67 & \cellcolor{blue!25}17.73 & 16.80 & \cellcolor{blue!25}41.27 & 109.93 & \cellcolor{blue!25}149.07 & 158.07 & \cellcolor{Gray}& \cellcolor{Gray}\\ 
 \cline{2-16}
  & RP & \cellcolor{blue!25}102.67 & 116.13 & \cellcolor{blue!25}151.73 & 177.67 & \cellcolor{blue!25}15.13 & 22.80 & \cellcolor{blue!25}16.93 & 21.60 & \cellcolor{blue!25}103.73 & 140.27 & \cellcolor{blue!25}166.87 & 188.27& \cellcolor{Gray}&\cellcolor{Gray} \\
 \hline
  \multirow{2}{*}{TO 4} & EP & \cellcolor{blue!25}59.13 & 67.20 & \cellcolor{blue!25}63.73 & 177.07 & \cellcolor{blue!25}11.27 & 16.27 & \cellcolor{blue!25}14.13 & 17.93 & \cellcolor{blue!25}42.53 & 82.87 & \cellcolor{blue!25}138.87 & 164.93& \cellcolor{Gray}&\cellcolor{Gray} \\ 
 \cline{2-16}
 & RP & \cellcolor{blue!25}119.87 & 141.73 & \cellcolor{blue!25}143.07 & 188.00 & \cellcolor{blue!25}13.33 & 21.00 & \cellcolor{blue!25}18.67 & 26.93 & \cellcolor{blue!25}95.07 & 125.53 & \cellcolor{blue!25}143.47 & 183.27 & \cellcolor{Gray}&\cellcolor{Gray} \\
  \hline
\end{tabular}
\end{table*}

\subsubsection{Complexity Comparison IPDB versus ISB}

To compare the relative computational complexity of IPDB and
ISB, Tables~\ref{tab:compl99} and \ref{tab:compl999} display the relative number of 
bit calculations (\ref{eq:bitrate}) for an 99\%
and 99.9\% accuracy, respectively, where the computational complexity of ISB is taken 
as a reference. Note that for the dual optimization part in ISB, the dual search
is optimized with tuned settings. As IPDB is a primal algorithm, it does
not have a dual optimization part, at the cost of more outer iterations.
However, it can be seen that IPDB has a much lower overal computational complexity for specific configuration settings. 
In particular the two-tone rand DoV transformation (\ref{eq:2toneDoVrand}) with 10~dBm/Hz granularity reduces complexity
by a factor 20 (for 99\% performance) and factor 12 (for 99.9\% performance).
For IPDB with the two-tone rand DoV transformation, EQ ON and 1~dBm/Hz granularity,
there is a 5\% to 50\% complexity reduction compared to ISB.

\begin{table*}
 \renewcommand{\arraystretch}{1.3}
 \caption{Relative complexity (wrt ISB standard method) to converge to 99\% performance for near-far CO-RT scenario of Fig.~\ref{fig:2adsl} with different settings for IPDB. Method Two-tone rand is with permutated tones on the transformation.
 Tone order (TO) 1 = [1:K], tone order 2 = [K:-1:1], tone order 3 = uniform probability of 1 and 2, tone order 4 = random permuation,
 EP = constant equal init power, RP = random init power satisfying total power constraint. Two right most columns correspond to ISB performance.}\label{tab:compl99}
\begin{tabular}{|c|c|c|c|c|c|c|c|c|c|c|c|c|c|c|c|}
 \hline
 \multicolumn{2}{|c|}{Method} & \multicolumn{4}{|c|}{Two-tone} & \multicolumn{4}{|c|}{Two-tone rand} & \multicolumn{4}{|c|}{Three-tone 2} & \multicolumn{2}{|c|}{ISB}\\
 \hline
 \multicolumn{2}{|c|}{Equalization} & \multicolumn{2}{|c|}{EQ OFF} & \multicolumn{2}{|c|}{EQ ON}& \multicolumn{2}{|c|}{EQ OFF} & \multicolumn{2}{|c|}{EQ ON}& \multicolumn{2}{|c|}{EQ OFF} & \multicolumn{2}{|c|}{EQ ON}& EQ OFF & EQ ON\\
 \hline
 \multicolumn{2}{|c|}{Granularity} & 1~dB & 10~dB & 1~dB & 10~dB & 1~dB & 10~dB & 1~dB & 10~dB & 1~dB & 10~dB & 1~dB & 10~dB & 0.5~dB & 0.5~dB\\
 \hline
  \multirow{2}{*}{TO 1} & EP & \cellcolor{blue!25}2.5164 & 0.2590 & \cellcolor{blue!25}7.6023 & 1.0543 & \cellcolor{blue!25}0.2739 & 0.0521 & \cellcolor{blue!25}0.4999 & 0.1178 & \cellcolor{blue!25}3.7865 & 0.3494 & \cellcolor{blue!25}5.1400 & 0.6725 & \cellcolor{LightCyan} 1.00 &  \cellcolor{LightCyan} 1.00\\ 
 \cline{2-16}
  & RP & \cellcolor{blue!25}1.7928 & 0.3406 & \cellcolor{blue!25}5.6130 & 1.0085 & \cellcolor{blue!25}0.3594 & 0.0792 & \cellcolor{blue!25}0.7443 & 0.0878 & \cellcolor{blue!25}1.4658 & 0.2554 & \cellcolor{blue!25}6.4844 & 1.0529 &  \cellcolor{LightCyan} & \cellcolor{LightCyan} \\
 \hline
  \multirow{2}{*}{TO 2} & EP &\cellcolor{blue!25}1.5272 & 0.3364 & \cellcolor{blue!25}3.0089 & 1.1248 & \cellcolor{blue!25}0.2723 & 0.0521 & \cellcolor{blue!25}0.4949 & 0.0521 & \cellcolor{blue!25}1.1267 & 0.3675 & \cellcolor{blue!25}5.0723 & 1.1433 & \cellcolor{Gray}& \cellcolor{Gray} \\ 
 \cline{2-16}
  & RP & \cellcolor{blue!25}2.0256 & 0.3157 & \cellcolor{blue!25}4.3902 & 0.8718 & \cellcolor{blue!25}0.3604 & 0.0741 & \cellcolor{blue!25}0.6630 & 0.1159 & \cellcolor{blue!25}1.3955 & 0.2603 & \cellcolor{blue!25}5.6175 & 1.0148 & \cellcolor{Gray} & \cellcolor{Gray}\\
 \hline
  \multirow{2}{*}{TO 3} & EP & \cellcolor{blue!25}1.6278 & 0.2804 & \cellcolor{blue!25}5.0953 & 1.1024 & \cellcolor{blue!25}0.2731 & 0.0521 & \cellcolor{blue!25}0.4973 & 0.0705 & \cellcolor{blue!25}1.4297 & 0.5400 & \cellcolor{blue!25}5.0541 & 0.9348 & \cellcolor{Gray}& \cellcolor{Gray}\\ 
 \cline{2-16}
  & RP & \cellcolor{blue!25}1.2612 & 0.2516 & \cellcolor{blue!25}4.9723 & 0.9191 & \cellcolor{blue!25}0.3991 & 0.0796 & \cellcolor{blue!25}0.4947 & 0.0765 & \cellcolor{blue!25}1.3478 & 0.2626 & \cellcolor{blue!25}5.8502 & 1.0582 & \cellcolor{Gray}&\cellcolor{Gray} \\
 \hline
  \multirow{2}{*}{TO 4} & EP & \cellcolor{blue!25}0.8795 & 0.2802 & \cellcolor{blue!25}2.2966 & 1.0672 & \cellcolor{blue!25}0.2729 & 0.0521 & \cellcolor{blue!25}0.4543 & 0.0706 & \cellcolor{blue!25}1.8095 & 0.3854 & \cellcolor{blue!25}6.1238 &  0.8910 & \cellcolor{Gray}&\cellcolor{Gray} \\ 
 \cline{2-16}
  & RP & \cellcolor{blue!25}1.9321 & 0.3380 & \cellcolor{blue!25}4.3290 & 1.0169 & \cellcolor{blue!25}0.3631 & 0.0740 & \cellcolor{blue!25}0.5791 & 0.1167 & \cellcolor{blue!25}1.2467 & 0.2398 & \cellcolor{blue!25}6.2019 & 1.0213 & \cellcolor{Gray}&\cellcolor{Gray} \\
  \hline
\end{tabular}
\end{table*}

\begin{table*}
 \renewcommand{\arraystretch}{1.3}
 \caption{Relative complexity (wrt ISB standard method) to converge to 99.9\% performance for near-far CO-RT scenario of Fig.~\ref{fig:2adsl} with different settings for IPDB. Method Two-tone rand is with permutated tones on the transformation.
 Tone order (TO) 1 = [1:K], tone order 2 = [K:-1:1], tone order 3 = uniform probability of 1 and 2, tone order 4 = random permuation,
 EP = constant equal init power, RP = random init power satisfying total power constraint. Two right most columns correspond to ISB performance.}\label{tab:compl999}
\begin{tabular}{|c|c|c|c|c|c|c|c|c|c|c|c|c|c|c|c|}
 \hline
 \multicolumn{2}{|c|}{Method} & \multicolumn{4}{|c|}{Two-tone} & \multicolumn{4}{|c|}{Two-tone rand} & \multicolumn{4}{|c|}{Three-tone 2} & \multicolumn{2}{|c|}{ISB}\\
 \hline
 \multicolumn{2}{|c|}{Equalization} & \multicolumn{2}{|c|}{EQ OFF} & \multicolumn{2}{|c|}{EQ ON}& \multicolumn{2}{|c|}{EQ OFF} & \multicolumn{2}{|c|}{EQ ON}& \multicolumn{2}{|c|}{EQ OFF} & \multicolumn{2}{|c|}{EQ ON}& EQ OFF & EQ ON\\
 \hline
 \multicolumn{2}{|c|}{Granularity} & 1~dB & 10~dB & 1~dB & 10~dB & 1~dB & 10~dB & 1~dB & 10~dB & 1~dB & 10~dB & 1~dB & 10~dB & 0.5~dB & 0.5~dB\\
 \hline
  \multirow{2}{*}{TO 1} & EP & \cellcolor{blue!25}4.2908 & 0.4292 & \cellcolor{blue!25}7.8222 & 1.1298 & \cellcolor{blue!25}0.5398 & 0.0945 & \cellcolor{blue!25}0.6694 & 0.1513 & \cellcolor{blue!25}3.9013 & 0.5441 & \cellcolor{blue!25}7.9446 & 0.6789 & \cellcolor{LightCyan} 1.00&  \cellcolor{LightCyan}1.00 \\ 
 \cline{2-16}
  & RP & \cellcolor{blue!25}4.5991 & 0.7979 & \cellcolor{blue!25}6.4685 & 1.1828 & \cellcolor{blue!25}0.6071 & 0.1217 & \cellcolor{blue!25}0.9525 & 0.1274 & \cellcolor{blue!25}3.7376 & 0.7171 & \cellcolor{blue!25}7.5192 & 1.1739 & \cellcolor{LightCyan} & \cellcolor{LightCyan} \\
 \hline
  \multirow{2}{*}{TO 2} & EP & \cellcolor{blue!25}1.6558 & 0.4217 & \cellcolor{blue!25}3.0948 & 1.2865 & \cellcolor{blue!25}0.4924 & 0.0945 & \cellcolor{blue!25}0.7036 & 0.0826 & \cellcolor{blue!25}2.1677 & 0.4326 & \cellcolor{blue!25}7.8462 & 1.2738 & \cellcolor{Gray}& \cellcolor{Gray} \\ 
 \cline{2-16}
  & RP &  \cellcolor{blue!25}4.7169 & 0.7045 & \cellcolor{blue!25}5.3924 & 1.0159 & \cellcolor{blue!25}0.5668 & 0.1118 & \cellcolor{blue!25}0.9126 & 0.1541 & \cellcolor{blue!25}3.2402 & 0.6662 & \cellcolor{blue!25}7.3323 & 1.1704 & \cellcolor{Gray} & \cellcolor{Gray}\\
 \hline
 \multirow{2}{*}{TO 3} & EP & \cellcolor{blue!25}2.3030 & 0.3722 & \cellcolor{blue!25}6.9218 & 1.2221 & \cellcolor{blue!25}0.4959 & 0.0945 & \cellcolor{blue!25}0.7904 & 0.1003 & \cellcolor{blue!25}1.9487 & 0.6497 & \cellcolor{blue!25}6.5151 & 1.0230 & \cellcolor{Gray}& \cellcolor{Gray}\\ 
 \cline{2-16}
  & RP & \cellcolor{blue!25}3.9985 &  0.6127 & \cellcolor{blue!25}6.6951 & 1.0937 & \cellcolor{blue!25}0.6448 & 0.1223 & \cellcolor{blue!25}0.7063 & 0.1223 & \cellcolor{blue!25}3.7474 & 0.7235 & \cellcolor{blue!25}7.5595 & 1.1591  & \cellcolor{Gray}&\cellcolor{Gray} \\
 \hline
  \multirow{2}{*}{TO 4} & EP & \cellcolor{blue!25}2.8825 & 0.4265 & \cellcolor{blue!25}3.3138 & 1.1623 & \cellcolor{blue!25}0.5368 & 0.1004 & \cellcolor{blue!25}0.6662 & 0.1063 & \cellcolor{blue!25}2.0686 & 0.4900 & \cellcolor{blue!25}6.6522 & 1.0611 & \cellcolor{Gray}&\cellcolor{Gray} \\ 
 \cline{2-16}
  & RP & \cellcolor{blue!25}4.9687 & 0.7453 & \cellcolor{blue!25}6.4080 & 1.1770 & \cellcolor{blue!25}0.5702 & 0.1170 & \cellcolor{blue!25}0.7883 & 0.1498 & \cellcolor{blue!25}3.6000 & 0.6725 & \cellcolor{blue!25}7.0006 & 1.1423 & \cellcolor{Gray}&\cellcolor{Gray} \\
  \hline
\end{tabular}
\end{table*}

\subsubsection{Real-time Property}\label{sec:realtimeproperty}

A main strength of IPDB is its real-time property. To compare with ISB, 
the number of power updates is determined that ISB maximally requires 
to satisfy the per-user total power constraints taking the inner 
dual optimization step into account.
Compared to one single power difference variable (and thus power) update for IPDB, 
ISB requires $1.5~10^6$ power updates to satisfy the per-user total power constraints (in worst case). This corresponds
to the worst case number of power updates (over all outer iterations) 
that a user needs to converge to transmit powers so as to satisfy its 
total power constraint. ISB thus does not qualify as a RT-DSM algorithm, in contrast to IPDB.

\subsubsection{Equalization Performance}\label{sec:equalperf}

To demonstrate the equalization impact, IPDB is simulated with RP, TO~4, 1~dBm/Hz granularity, two-tone rand DoV transformation and
equalization procedure Algorithm~\ref{algo:rtso_equal} activated (EQ ON) 
for every outer iteration which is an integer multiple of 5.

The resulting transmit spectra before and after equalization are shown
for iteration 5 and 10 in Figure~\ref{fig:idbloops}. It can be seen that
after the equalization step the transmit spectra display fewer spikes.
In this case, only after two equalization steps, all spikes are removed, 
demonstrating the effectiveness of the equalization procedure.

 \begin{figure*}[!t]
  \centering
  \includegraphics[width=1.85\columnwidth]{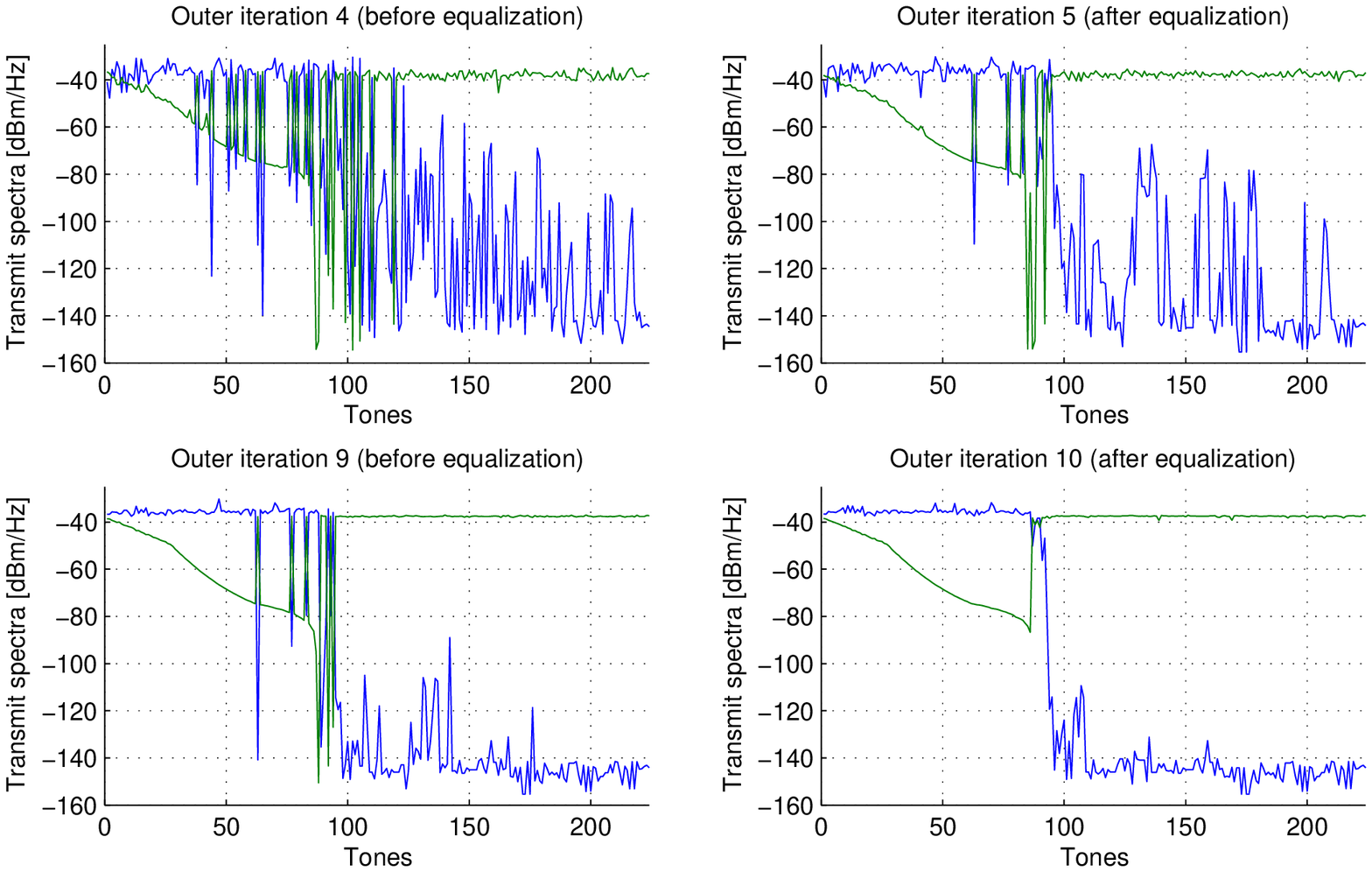}
  \caption{Impact when equalization procedure Algorithm~\ref{algo:rtso_equal} is
  activated for IPDB at outer iteration 5 and 10, before and after the equalizations. Transmit spectra of CO-user and RT-user in Figure~\ref{fig:2adsl} are displayed in 
  blue and green, respectively.}
  \label{fig:idbloops}
 \end{figure*}

\subsection{Impact of Discretization Granularity: ADSL2+ case}\label{sec:gran}

As mentioned in Section~\ref{sec:idb}, the discrete
grid-based search granularity for the power difference variables for IPDB can be chosen
coarser than for the transmit powers for ISB, i.e., 1~dBm/Hz instead of 0.5~dBm/Hz. 
In this section we assess the concrete impact of different granularities for an ADSL2+ scenario
as given in Figure~\ref{fig:adsl2plusscenario}, in terms final transmit spectra. The downstream ADSL2+ scenario consists of 12 users with line lengths
5000m, 4000m, 3000m, 2000m, 2000m, 1000m, 4800m, 3800m, 2800m, 2300m, 1500m, and 1300m.
The distances (between CO and RTs) are 0m, 0m, 1000m, 1000m, 2000m, 2000m, 0m, 0m, 1200m, 1200m, 2400m, and 2400m.

\begin{figure}[!t]
  \centering
  \includegraphics[width=0.75\columnwidth]{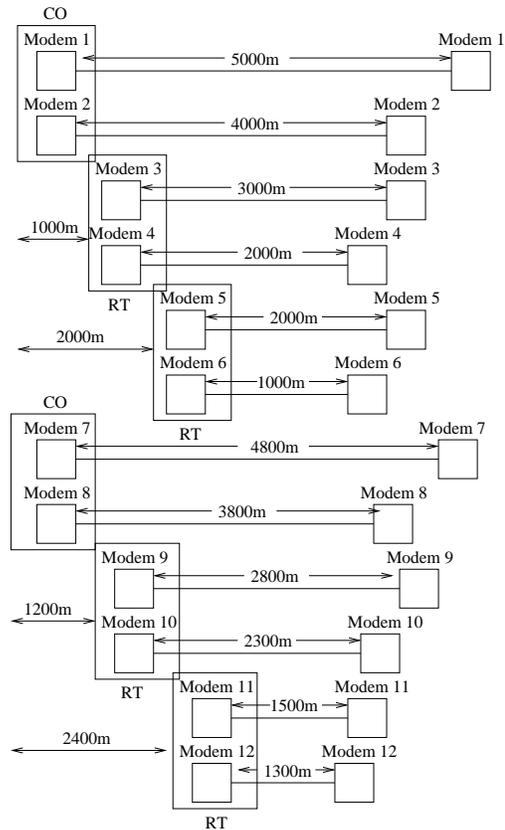}
  \caption{ADSL2+ scenario}
  \label{fig:adsl2plusscenario}
 \end{figure}

We compare three configurations with each other: a) IPDB with 1~dBm/Hz granularity,
b) IPDB with 10~dBm/Hz granularity, and c) ISB with 0.5~dBm/Hz granularity.
For the three configurations, we start from the same initial transmit powers
(EP setting) and all three converge to similar solutions. The weighted sum achievable data rate
performance corresponds to: a) 16.0072~Mbps, b) 15.9960~Mbps, c) 16.0055~Mbps.

In Figure~\ref{fig:granularity}, the resulting transmit spectra for user 7 are zoomed out
for tones 230 to 420 for the three above configurations. It can be seen that ISB makes
steps of 0.5~dBm/Hz. In contrast, the IPDB methods (for both granularities) display 
a smaller step variation in magnitude, which demonstrates that a much coarser
granularity for IPDB does not impact the shape of the resulting transmit spectra too much.
There is some level of non-smooth behaviour for IPDB though. The equalization 
procedure is not able to remove this, because the equalization 
threshold is set at 10~dBm/Hz (as shown in lines 3,4,6,7 in Algorithm~\ref{algo:rtso_equal}).

\begin{figure}[!t]
\centering
\includegraphics[width=0.95\columnwidth]{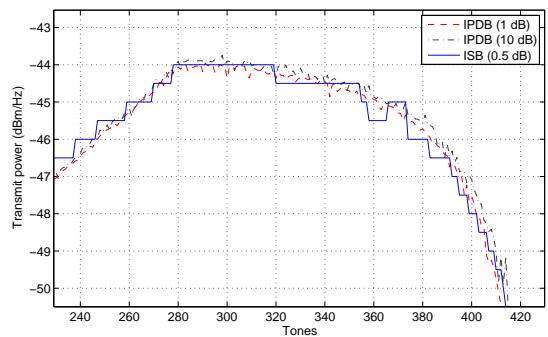}
\caption{Zoom of transmit spectra for user 7 of 12-user ADSL2+ scenario}
\label{fig:granularity}
\end{figure}

\subsection{Inequality Constraints: VDSL case}\label{sec:inequalsim}

For typical DSL scenarios, all users are allocated their full available transmit power satisfying
the total power constraints with equality, i.e., $P^n = P^{n,\mathrm{tot}}$. 
However, for multi-user large crosstalk settings and under specific values for the weights $w_n$, 
it is possible that some users better not be allocated all available per-user total power. 
Here we consider a 6-user VDSL upstream scenario, with 6 CO-connected lines
with line lengths 1200m, 1000m, 800m, 600m, 450m, and 300m, corresponding 
to users 1 to 6, respectively.
For this scenario, users 4 and 5 are not allocated all available per-user total power, i.e.,
$P^n < P^{n,\mathrm{tot}}$. When running the IPDB algorithm (with RP, EQ ON,
1~dBm/Hz granularity, TO 4, DoV transformation (\ref{eq:2toneDoVrand})), with the inequality procedure
of Algorithm~\ref{algo:rtso_inequal} (with  $\alpha=1.1$ and $\beta = 0.8$), the evolution of the allocated per-user total 
powers is shown in Figure~\ref{fig:inequality}. 
It can be seen that the per-user total power constraints are always satisfied. After 30 iterations, both
users achieve their final per-user total power allocation corresponding to 23\% and  1.7\% of full power allocation $P^{n,\mathrm{tot}}$.
This demonstrates the real-time property of IPDB while considering inequality constraints.

\begin{figure}[!t]
\centering
\includegraphics[width=0.95\columnwidth]{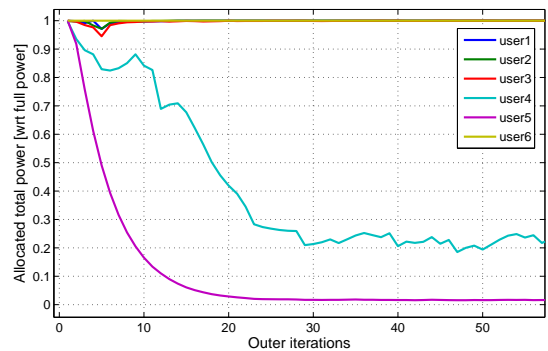}
\caption{Evolution of user powers for IPDB with inequality procedure of Algorithm~\ref{algo:rtso_inequal} for 6-user upstream VDSL scenario.}
\label{fig:inequality}
\end{figure}

\subsection{Downlink Power Control in Heterogeneous Wireless Networks}\label{sec:lte}

As highlighted in Section~\ref{sec:introduction}, the proposed RT-DSM theory and IPDB
algorithm can also be applied to wireless communication settings. One highly relevant problem is downlink power control in heterogeneous
cellullar networks where OFDMA is used within each cell and inter-cell
interference is observed between different (macro, pico, femto) cells.
In \cite{Son2011} it is explained that this consists of two subproblems, namely
a user scheduling part and a power spectrum control part. For the power 
spectrum control part, it is shown how DSM algorithms can lead to spectacular
performance gains. The IPDB algorithm can similarly be applied to this
setting so as to obtain real-time inter-cell interference coordination (ICIC)
for such heterogeneous networks.

The IPDB algorithm is applied here to a system with two interfering cells, one
macrocell and one femtocell. Each cell has one user. The user in the macrocell
is located at the cell edge at a distance of 500m from the macrocell base station 
and 20m from the femtocell base station. The user in the femtocell is located
at the same location but is served by the femtocell base station. It is 
known that this constitutes a challenging interference limited setting. 
We consider a system bandwidth of 5~MHz, a subcarrier spacing of 15~kHz, 
a symbol rate of 14 OFDM symbols in 1ms, 300 subcarriers, a macrocell base 
station transmit power of 43 dBm and a femtocell base station transmit power 
of 15 dBm. The ITU-PED B channel model is used with a pathloss of 
$31.5 + 35 \log_{10}$(distance).

The resulting bit loadings are displayed in Figure~\ref{fig:bitLTE}. As the
scenario deals with two edge users, the resulting transmit power allocations
are OFDMA like. There is a small overlap though in tones 20-22. The goal is however not
to analyze the resulting transmit spectra and bit loadings, but is to  
demonstrate the wide applicability of the proposed RT-DSM theory and algorithm
beyond the wireline DSL setting.

\begin{figure}[!t]
\centering
\includegraphics[width=0.95\columnwidth]{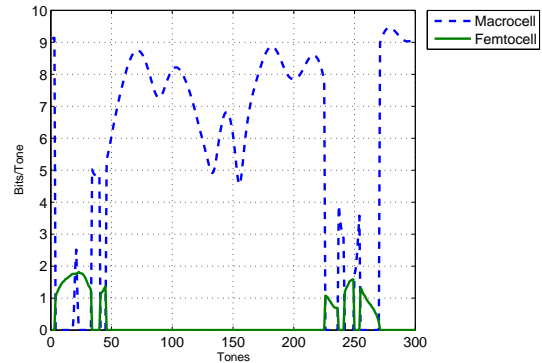}
\caption{Bit allocation with IPDB for LTE heterogeneous network with macro- and
femtocell inter-cell interference.}
\label{fig:bitLTE}
\end{figure}

\section{Conclusion}
We have proposed a new paradigm, theory and algorithm for RT-DSM in multi-user multi-carrier communication systems. 
The RT-DSM algorithm is referred to as IPDB. IPDB is suitable for operation
under tight computation time and compute power constraints, i.e., when a very
fast responsiveness is required. IPDB can be stopped at any moment in time during execution 
while guaranteeing feasibility and improved performance. The IPDB algorithm
design is enabled by a novel transformation, referred to as the DoV transformation, which transforms the standard DSM 
problem into an alternative optimization problem with primal power difference variables.
A coordinate ascent approach is proposed to tackle the reformulated primal problem
with an iterative 1D exact line search via a logarithmicly scaled grid search.
In contrast to existing DSM algorithms that follow a dual
decomposition approach, IPDB solves the DSM problem in 
the primal domain, avoiding any potential issues with a non-zero duality gap, which
can be seen as an important benefit. IPDB is furthermore characterized by a high tunability with additional 
procedures for dealing with inequalities and equalization that result in
improved performance and smooth transmit power spectra. In particular, the 
configuration with the 'two-tone rand' DoV transformation and equalization results in fast 
convergence, good network wide achievable data rate performance, low computational cost, 
and real-time execution, outperforming the existing near-optimal ISB algorithm. This has been validated with simulations 
under different configuration settings for different practical wireline xDSL
scenarios and for a wireless LTE heterogeneous network scenario.

% or


% \appendices
% % you can choose not to have a title for an appendix
% % if you want by leaving the argument blank
% \section{}
% % use section* for acknowledgement
% \section*{Acknowledgment}


\ifCLASSOPTIONcaptionsoff
  \newpage
\fi



%\bibliographystyle{IEEEtran}
% \end{thebibliography}

\bibliographystyle{IEEEtran}
\bibliography{bibliographyshort}

%\begin{biography}[{\includegraphics[width=1in,height=1.25in,clip,keepaspectratio]{mshell}}]{Michael Shell}

% \end{IEEEbiography}
% % insert where needed to balance the two columns on the last page with
% %\newpage


%\vfill


\end{document}